\documentclass{article}
\usepackage[paper=a4paper,left=25mm,right=25mm,top=25mm,bottom=25mm]{geometry}

\usepackage{amsmath, amsfonts, mathtools, amsthm}
\usepackage{authblk}
\usepackage{enumitem}
\usepackage{bbm}
\usepackage{tabularx}
\usepackage{url}
\usepackage{graphicx}

\usepackage{bibunits}
\defaultbibliographystyle{plain} 
\defaultbibliography{oslr_ml_variability_bib}

\allowdisplaybreaks

\newcommand{\boldtheta}{\boldsymbol{\theta}}
\newcommand{\boldTheta}{\boldsymbol{\Theta}}
\newcommand{\tmax}{t_{\text{max}}}

\newtheorem{theorem}{Theorem}
\newtheorem{lemma}{Lemma}
\newtheorem{corollary}{Corollary}

\title{Correcting for sampling variability in maximum likelihood-based one-sample log-rank tests}

\author[1,*]{Moritz Fabian Danzer}
\author[1]{Rene Schmidt}

\affil[1]{Institute of Biostatistics and Clinical Research, University of Münster, Germany}
\affil[*]{Corresponding author, mail: moritzfabian.danzer@ukmuenster.de}

\usepackage{hyperref}

\begin{document}
\begin{bibunit}

\maketitle

\begin{abstract}
	Single-arm studies in the early development phases of new treatments are not uncommon in the context of rare diseases or in paediatrics. If an assessment of efficacy is to be made at the end of such a study, the observed endpoints can be compared with reference values that can be derived from historical data. For a time-to-event endpoint, a statistical comparison with a reference curve can be made using the one-sample log-rank test. In order to ensure the interpretability of the results of this test, the role of the reference curve is crucial. This quantity is often estimated from a historical control group using a parametric procedure. Hence, it should be noted that it is subject to estimation uncertainty. However, this aspect is not taken into account in the one-sample log-rank test statistic. We analyse this estimation uncertainty for the common situation that the reference curve is estimated parametrically using the maximum likelihood method, and indicate how the variance estimation of the one-sample log-rank test can be adapted in order to take this variability into account. The resulting test procedures are illustrated using a data example and analysed in more detail using simulations, particularly in comparison with established two-sample methods.
\end{abstract}

\section{Introduction}\label{sec:introduction}
Single-arm clinical trials are still frequently conducted as part of the development process of a new treatment. They are used in particular in earlier phases in which the initial effectiveness of the treatment is to be determined (Phase II), but also in general in the case of rare diseases. If the primary endpoint is a time-to-event endpoint, the one-sample log-rank test, which has been introduced in \cite{Breslow:1975}, is used. The primary endpoint data collected in the study is compared with a reference curve. This curve can be determined based on historic data, for example registers or data from previous studies that reflect the current standard of care. In particular, the choice arises as to whether the reference curve should be estimated parametrically or non-parametrically from the available data.\\
In both cases, the one-sample log-rank test has to be viewed critically when it comes to the role of the reference curve. To ensure that the comparison is still clinically interpretable, it must first be ensured that the distribution of important baseline characteristics in the historic data corresponds to that expected for the new cohort to be recruited. In addition, there is a risk that the historic data might not reflect recent advances in diagnostics and/or concomitant therapy for standard of care. A careful selection of the historical cohort, taking into account the relevant literature (see e.g. \cite{Pocock:1976, Ghadessi:2020}), is urgently needed to address these issues. From a statistical point of view, the sampling variability of the estimation process of the reference curve from the historical data should be scrutinised intensively. This uncertainty is ignored in the one-sample log-rank test as it is assumed that the reference is fixed and deterministic. As we will also point out below, the resulting test addresses a different hypothesis from the one in which the researcher is usually interested.\\
This issue has been addressed by \cite{Danzer:2022} in the case where the reference curve is estimated by a non-parametric procedure from the historical control data, i.e. by the Kaplan-Meier or the Nelson-Aalen estimator for the survival and the cumulative hazard function, respectively. In this manuscript, we want to investigate the effects on the properties of the test when the reference curve is estimated parametrically. A large number of distribution families are available for this purpose (see e.g. Section 2.5 of \cite{Klein:2005}), which are also implemented in software (e.g. the R package \texttt{flexsurv} \cite{Jackson:2016}). In \cite{Szurewsky:2025}, several application examples have been presented. As performed there and also suggested in Chapter 12.4 of \cite{Klein:2005}, criteria such as the Akaike information criterion (AIC) can be used to find the best model among the possible candidates. For these parametric methods of deriving the reference curve, we would now also like to analyse how this procedure affects the variability of the test statistics of the one-sample log-rank test.\\
We first introduce some notation in Section \ref{sec:notation} and briefly discuss the hypotheses already mentioned, which are of interest in this context and are to be tested statistically. In \ref{sec:par_est}, we introduce some basic quantities and results of maximum likelihood estimation. We require them in order to quantify the effects of the estimation procedures on the statistical testing methods for the different hypotheses in Section \ref{sec:testing}. Afterwards, we demonstrate the different procedures in an example based on real data (Section \ref{sec:application}) and investigate small sample properties of the asymptotically valid procedures in a simulation study (Section \ref{sec:sim_study}). Finally, the results are discussed in Section \ref{sec:discussion}. Further simulation results and a complete derivation of the asymptotic results can be found in the Supplementary Material. The code underlying the case study and the simulation study is available at \url{https://github.com/moedancer/MLEVar_OSLR}.

\section{Notation}\label{sec:notation}
As in \cite{Danzer:2022}, we assume that historic data on the control treatment (e.g. the current standard of care) is available. In addition, there is further data from a single-arm study in which patients undergo an experimental therapy that is to be compared with the control therapy. These two groups are referred to below as group A and group B, respectively.\\ 
Let $\mathcal{N}_x$ denote the set of patients from group $x \in \{A,B\}, \; n_x \coloneqq |\mathcal{N}_x|$ the number of patients in the two groups and $n\coloneqq n_A + n_B$ the total number of patients. We denote the ratio between the two group sizes as $\pi=n_B/n_A$.\\
We are interested in the distribution of a time-to-event variable $T$. We assume that the observations in both groups are distributed as the two random variables $T_A$ and $T_B$, respectively. However, observation of this variable might be censored due to an independent censoring mechanism through the random variable $C_A$ in the historic control group or $C_B$ in the prospective group. In any case, we can only observe $X_x \coloneqq \min(T_x,C_x)$ and $\delta_x \coloneqq \mathbbm{1}(T_x\leq C_x)$ for $x \in \{A,B\}$. We assume that all patients are observed for a maximum time $\tmax < \infty$, i.e. $X\leq \tmax$ with probability 1. Density, cumulative distribution, survival, hazard and cumulative hazard functions will be denoted by $f,F,S,\lambda$ and $\Lambda$, respectively. The variable to which the function refers is specified in the index. For example, $f_{T_A}$ denotes the density function of the random variable $T_A$, i.e. the time-to-event variable in the historical control cohort. For $x \in \{A,B\}$, we assume that $n_x$ independent and identically distributed replicates of the tuple $(X_x, \delta_x)$ are given which we will denote by $(X_{x,i}, \delta_{x,i})$ for $i\in \{1,\dots,n_x\}$. We also assume that the data from both groups are independent of each other. By the functions $N_{x,i}, Y_{x,i} \colon \mathbb{R}_+ \to \{0,1\}$ we denote the time-dependent event and at-risk indicators for patient $i \in \{1,\dots,n_x\}$ in group $x \in \{A,B\}$, i.e.
\begin{equation*}
N_{x,i}(t) \coloneqq \mathbbm{1}(T_{x,i} \leq t \wedge C_{x,i}) \qquad \text{and} \qquad Y_{x,i}(t) \coloneqq \mathbbm{1}(X_{x,i} \geq t)
\end{equation*}
for all $t \geq 0$. By $N_x$ and $Y_x$ we denote the accumulated quantities over the whole group.\\
When comparing the data from these two groups, we are of course interested in the examination of the null hypothesis
\begin{equation}\label{eq:H0}
	H_0\colon \Lambda_{T_B}(t) = \Lambda_{T_A}(t) \quad \text{for all }t\in[0,\tmax].
\end{equation}
In practice, such a comparison is carried out by comparing the prospectively collected data of group B with a reference curve derived from the historic data of group A. Let this reference curve in terms of the cumulative hazard function be given by $\Lambda_{T_A, \text{ref}}$. The corresponding test, which is called the one-sample log-rank test, investigates the hypothesis
\begin{equation*}
	H_{0,\text{OSLR}}\colon \Lambda_{T_B}(t) = \Lambda_{T_A, \text{ref}}(t) \quad \text{for all }t\in[0,\tmax].
\end{equation*}
As this test only makes a comparison with the reference curve, the uncertainty inherent in its determination is not taken into account.  Hence, if we interpret the one-sample log-rank test as a test of the hypothesis $H_0$ from \eqref{eq:H0}, the type I error rate can be much higher than the nominal level of the one-sample log-rank test. The function $\Lambda_{T_A, \text{ref}}$ can be determined non-parametrically by the Nelson-Aalen estimator. In such settings, the type I error rate inflation that occurs when the one-sample log-rank test is regarded as a test of $H_0$ is determined in \cite{Danzer:2022}. Based on this, a test statistic is proposed in \cite{Feld:2024} that can be calculated solely from summary statistics from group A.\\ 
We will now focus on the case where $\Lambda_{T_A, \text{ref}}$ is estimated parametrically from the control group data. Our purpose is to provide a corrected one-sample log-rank test for this setting, too. Therefore, we consider a parametrised family of distributions of the variable $T_A$. The parameter space $\boldTheta$ is an open subset of the $q$-dimensional Euclidean space and we assume $T_A$ to be distributed according to the element $\boldtheta_A \in \boldTheta$. The functions mentioned above describing the distribution of $T_A$ will now be functions of two arguments and hence map from $\boldTheta \times \mathbb{R}_+$. If $\hat{\boldtheta}_A$ denotes the estimated parameter from historic control data, the reference curve will thus be determined by $\Lambda_{T_A, \text{ref}} = \Lambda_{T_A}(\hat{\boldtheta}_A, \cdot)$.\\
Our aim here is to quantify the inflation of the type I error level of the one-sample log-rank test when a one-sample log-rank test with $\Lambda_{T_A, \text{ref}} = \Lambda_{T_A}(\hat{\boldtheta}_A, \cdot)$ is used to test $H_0$ instead of $H_{0,\text{OSLR}}$ and to adjust the test in such a way that it can be used as a valid test for $H_0$. 

\section{Parametric estimation}\label{sec:par_est}
A first look at the impact of the estimation variability on the type I error rate of the one-sample log-rank test has already been taken in \cite{Szurewsky:2025} via simulation. There, however, only exponentially distributed variables were considered, whereby the parameter of the exponential distribution was derived via the median survival time in the historical control cohort. We now want to cover as large a number of distributions as possible, whereby the parameters are estimated using the maximum likelihood procedure. This estimation is probably the most common method. For our considerations, the well-known asymptotic results for this method are also extremely relevant in order to quantify the effects on the test statistics. We will first briefly discuss these results and the necessary conditions.\\
For maximum likelihood estimation of the parameter $\boldtheta_A = (\theta_{A,1}, \dots, \theta_{A,q})$ we require the standard regularity conditions to hold s.t. we have consistency of the maximum likelihood estimator $\hat{\boldtheta}_A$, i.e. $\hat{\boldtheta}_A \overset{\mathbb{P}}{\to} \boldtheta_A$ and asymptotic normality of this estimator, i.e.
\begin{equation}
\sqrt{n_A} \left( \hat{\boldtheta}_A - \boldtheta_A \right) \overset{\mathcal{D}}{\to} \mathcal{N}\left( 0, \mathcal{I}(\boldtheta_a)^{-1} \right)
\end{equation}
where $\mathcal{I}(\boldtheta_A)$ is the expected Fisher information matrix which is a $q \times q$-matrix with entries
\begin{equation}
\mathcal{I}_{kl}(\boldtheta_A) = \mathbb{E}\left[ -\frac{\partial^2}{\partial \theta_k \partial \theta_l} l(\boldtheta, (X_{A}, \delta_{A})) \Big|_{\boldtheta=\boldtheta_A} \right]
\end{equation}
with the log-likelihood $l$ given by
\begin{equation}
l(\boldtheta, (X, \delta))\coloneqq 
\begin{cases}
\log f_{T}(\boldtheta, X)&\text{if }\delta=1\\
\log S_{T}(\boldtheta, X)&\text{else}.
\end{cases}
\end{equation}
These conditions can be found e.g. in Section VI.1.2 of \cite{Andersen:1993}. The consistency and asymptotic normality of the estimator itself and the consistency of the estimator of the covariance matrix follow from Theorems VI.1.1 and VI.1.2 of \cite{Andersen:1993}. The proofs thereof follow the lines of \cite{Borgan:1984}.\\
The empirical counterpart $\mathcal{J}$ of the Fisher information matrix is given by its components
\begin{align*}
	\mathcal{J}_{kl}(\boldtheta) \coloneqq -\frac{1}{n_A} \sum_{i \in \mathcal{N}_A} \left( \int_0^{t_{\max}} \frac{\partial^2}{\partial \theta_k \partial \theta_l} \log \lambda(\boldtheta, u) dN_{A,i}(u) -  \int_0^{t_{\max}} \frac{\partial^2}{\partial \theta_k \partial \theta_l} \lambda(\boldtheta, u) Y_{A,i}(u) du \right)
\end{align*}
for $k,l \in \{1,\dots,q\}$. For our purposes, we explicitly require the following:
\begin{enumerate}[label = (A\arabic*)]
	\item \label{item:continuity} The function $\Lambda \colon \Theta \times \mathbb{R}_+ \to \mathbb{R}_+$ is continuous
	\item \label{item:differentiability} The function $\boldtheta \mapsto \Lambda(\boldtheta,s)$ is continuously differentiable at $\boldtheta_A$ for any $s\in [0,\tmax]$.
	\item \label{item:exchange_expectation_differentiation} There exists a real-valued function $g$ with $||\nabla_{\boldtheta} \Lambda(\boldtheta, s)|| \leq g(s)$ for all $\boldtheta \in \overline{B_{\delta}}(\boldtheta_A)$ for some $\delta > 0$ and all $s \in (0,\tmax)$ s.t. $\mathbb{E}[g(t\wedge X_B)] < \infty$. 
\end{enumerate}
These assumptions are fulfilled for common parametric models, which we will consider here. As already mentioned information criteria as the AIC can be used to determine the best candidate among possible distributions. The AIC is given by
\begin{equation*}
	2\cdot q - 2 \cdot \sum_{i \in \mathcal{N}_A} l(\hat{\boldtheta}_A, (X_{A,i}, \delta_{A,i})).
\end{equation*}
The candidate with the lowest AIC value is usually chosen. Nevertheless, other selection criteria and a visual inspection of the fit can also play a major role in the selection process.

\section{Testing procedures}\label{sec:testing}
The hypothesis $H_{0,\text{OSLR}}$ can basically be tested by subtracting the expected number of events before some time $t$ from the observed number of events before that time $t$ which is given by $N_B(t)$. The number of expected events under the assumption that $\Lambda_{T_A, \text{ref}}$ is the true cumulative hazard function in group B is given by
\begin{equation*}
	E_{B,\text{OSLR}}(t) \coloneqq \sum_{i \in \mathcal{N}_B} \Lambda_{T_A, \text{ref}}(X_{B,i} \wedge t) = \int_0^t Y_B(s)\, d\Lambda_{T_A, \text{ref}}(s).
\end{equation*}
Under $H_{0,\text{OSLR}}$, the stochastic process $(M_{\text{OSLR}}(t))_{t \geq 0}$ defined by $M_{\text{OSLR}}(t) \coloneqq n_B^{-1/2} (N_B(t) - E_{B,\text{OSLR}}(t))$ is a martingale w.r.t. its natural filtration. Based on a central limit theorem for martingales (see Theorem II.5.1 of \cite{Andersen:1993}), the value at some time $t$ of this process asymptotically follows a normal distribution with mean 0 and variance 
\begin{equation*}
	V_1(t) \coloneqq \mathbb{P}_{H_{0,\text{OSLR}}}[T_B \leq C_B \wedge t] = \mathbb{E}_{H_{0,\text{OSLR}}}[\Lambda_{T_A, \text{ref}}(X_B \wedge t)]
\end{equation*}
under the null hypothesis $H_{0,\text{OSLR}}$. As $V_1(t)$ can be consistently estimated by $n_B^{-1} \cdot N_B(t)$ as well as by $n_B^{-1} \cdot E_{B,\text{OSLR}}(t)$, we can estimate $V_1$ consistently by
\begin{equation}\label{eq:oslr_var}
	\hat{V}_1(w,t)\coloneqq \frac{1}{n_B}\left(w \cdot N_B(t) + (1-w) \cdot E_{B,\text{OSLR}}(t)\right)
\end{equation}
for any $w \in [0,1]$. In the original one-sample log-rank test, $w=0$ has been chosen implicitly \cite{Breslow:1975}. Other choices have been discussed by \cite{Wu:2014} who choose $w=0.5$ and \cite{Danzer:2023} who propose an approach that takes the anticipated censoring mechanism into account.\\ 
In order to use all available data to test $H_{0,\text{OSLR}}$ we can use the test statistic
\begin{equation*}
	Z_{\text{OSLR}} \coloneqq \frac{N_B(t_{\max}) - E_{B,\text{OSLR}}(t_{\max})}{\sqrt{w\cdot N_B(t_{\max}) + (1-w) \cdot E_{B,\text{OSLR}}(t_{\max})}}
\end{equation*}
for some $w \in [0,1]$. Asymptotically, the choice of $w$ is not relevant and traditionally, a value of $w=0$ has been chosen \cite{Breslow:1975}. However, small sample properties can be improved by choosing $w=0.5$ as in \cite{Wu:2014} or adapting the choice to the anticipated censoring pattern \cite{Danzer:2023}. By the asymptotic normality and the consistency of the estimators of $V_1$, an asymptotically correct two-sided test with type I error level $\alpha$ of $H_{0,\text{OSLR}}$ is given by rejecting it if
\begin{equation*}
	|Z_{\text{OSLR}}| \geq \Phi^{-1}\left( 1 - \frac{\alpha}{2} \right).
\end{equation*} 
When the cumulative reference hazard function $\Lambda_{T_A, \text{ref}}$ is determined by a parametric estimate from historic data, however, the asymptotic considerations made above and the corresponding test statistic do not take the resulting estimation uncertainty into account.\\
However, there is usually a higher interest in the investigation of $H_0$. Unfortunately, $H_{0,\text{OSLR}}$ and $H_0$ do not coincide in general as $\Lambda_{T_A, \text{ref}}$ only attempts to mimic $\Lambda_{T_A}$. If a family of parametric distributions appears appropriate, the emulation of $\Lambda_{T_A, \text{ref}}$ by $\Lambda_{T_A}(\hat{\boldtheta}_A, \cdot)$ often employs a maximum likelihood estimate of the parameter $\boldtheta_A$ from the historic data of group A. In this case $\Lambda_{T_A, \text{ref}}$ will be replaced by $\Lambda_{T_A}(\hat{\boldtheta}_A,\cdot)$. But even if the parametric family is correct, the estimated parameter $\hat{\boldtheta}_A$ will be subject to random variability around the true value. This additional variability has to be taken into account when estimating the variance of $M_{\text{OSLR}}(t)$ under the more general hypothesis $H_0$. It is shown in the Supplementary Material that 
\begin{equation*}
	M_{\text{OSLR}}(t) \overset{\mathcal{D}}{\to} \mathcal{N}(0, V_1(t) + V_2(t))
\end{equation*}
under $H_0$. As this formula suggests, we can decompose the total variance under $H_0$ into the part that constitutes the variance $V_1$ under $H_{0,\text{OSLR}}$ and some additional variance $V_2$ that is given by
\begin{equation}\label{eq:add_var}
V_2(t)\coloneqq \pi \cdot\mathbb{E}[\nabla_{\boldtheta} \Lambda(\boldtheta_A, t \wedge X_B)]^T \mathcal{I}(\boldtheta_A)^{-1} \mathbb{E}[\nabla_{\boldtheta} \Lambda(\boldtheta_A, t \wedge X_B)]
\end{equation}
where $\nabla_{\boldtheta}$ denotes the gradient w.r.t. the parameter of the distribution. Similar to the results of \cite{Danzer:2022} we see a clear dependence from the allocation ratio $\pi$. In particular for fixed $n_B$ we can see that the variance component $V_2$ vanishes as the size of the historic control cohort $n_A$ increases. This variance component can also be estimated consistently by
\begin{equation}\label{eq:mle_var}
	\hat{V}_2(t) \coloneqq \pi \cdot \left(\frac{1}{n_B} \sum_{i=1}^{n_B} \nabla_{\boldtheta} \Lambda(\hat{\boldtheta}_A, t \wedge X_{B,i}) \right)^T \mathcal{J}(\hat{\boldtheta}_A)^+ \left( \frac{1}{n_B} \sum_{i=1}^{n_B} \nabla_{\boldtheta} \Lambda(\hat{\boldtheta}_A, t \wedge X_{B,i}) \right)
\end{equation}
where $\mathcal{J}(\hat{\boldtheta}_A)^+$ denotes the Moore-Penrose inverse of the estimator of the Fisher information matrix. This allows us to construct a corrected test that maintains the aspired type I error rate under $H_0$. This test statistic is given by
\begin{equation*}
	Z_{\text{corr}} \coloneqq \frac{n_B^{-\frac{1}{2}}\left(N_B(t) - \sum_{i \in \mathcal{N}_B} \Lambda(\hat{\boldtheta}_A, t \wedge X_{B,i})\right)}{\sqrt{\hat{V}_1(w,t) + \hat{V}_2(t)}}.
\end{equation*}
Hence, we would reject $H_0$ at the two-sided significance level of $\alpha$ if
\begin{equation*}
	|Z_{\text{corr}}| \geq \Phi^{-1}\left( 1 - \frac{\alpha}{2} \right).
\end{equation*}
Corresponding two-sided $p$-values $p_{\text{OSLR}}$ of the original one-sample log-rank test and $p_{\text{corr}}$ of the corrected test are thus given by
\begin{equation*}
	p_{\text{OSLR}} \coloneqq 2 \cdot (1 - \Phi ( |Z_{\text{OSLR}}| )) \quad \text{and} \quad  p_{\text{corr}} \coloneqq 2 \cdot (1 - \Phi ( |Z_{\text{corr}}| )),
\end{equation*}
respectively. Since $\hat{V}_1(w,t) + \hat{V}_2(t) > \hat{V}_1(w,t)$, we will always obtain $p_{\text{OSLR}} < p_{\text{corr}}$ for the two-sided p-values making the difference between the one-sample log-rank statistic and our corrected version once more obvious. The one-sided p-values are given by
\begin{equation*}
	p_{\text{OSLR}} \coloneqq \Phi ( Z_{\text{OSLR}} ) \quad \text{and} \quad  p_{\text{corr}} \coloneqq \Phi ( Z_{\text{corr}} )
\end{equation*}
where smaller p-values correspond to a larger indication of superiority of the experimental group.

\section{Application example}\label{sec:application}
In this section, we want to demonstrate our correction procedure in a practical example based on real data. For this purpose, we consider the same example as in Section 7.2 of \cite{Szurewsky:2025}. It is based on the data of a phase II trial in paediatric patients with relapsed or refractory neuroblastoma. In this single-arm trial, the inhibitor ABT-751 was evaluated. The results in this experimental arm were compared to that of a historic control group from previous studies. Among several other endpoints, also overall survival was evaluated. The results can be found in \cite{Fox:2014}.\\
\begin{figure}
	\centering
	\includegraphics[width = .8\textwidth]{"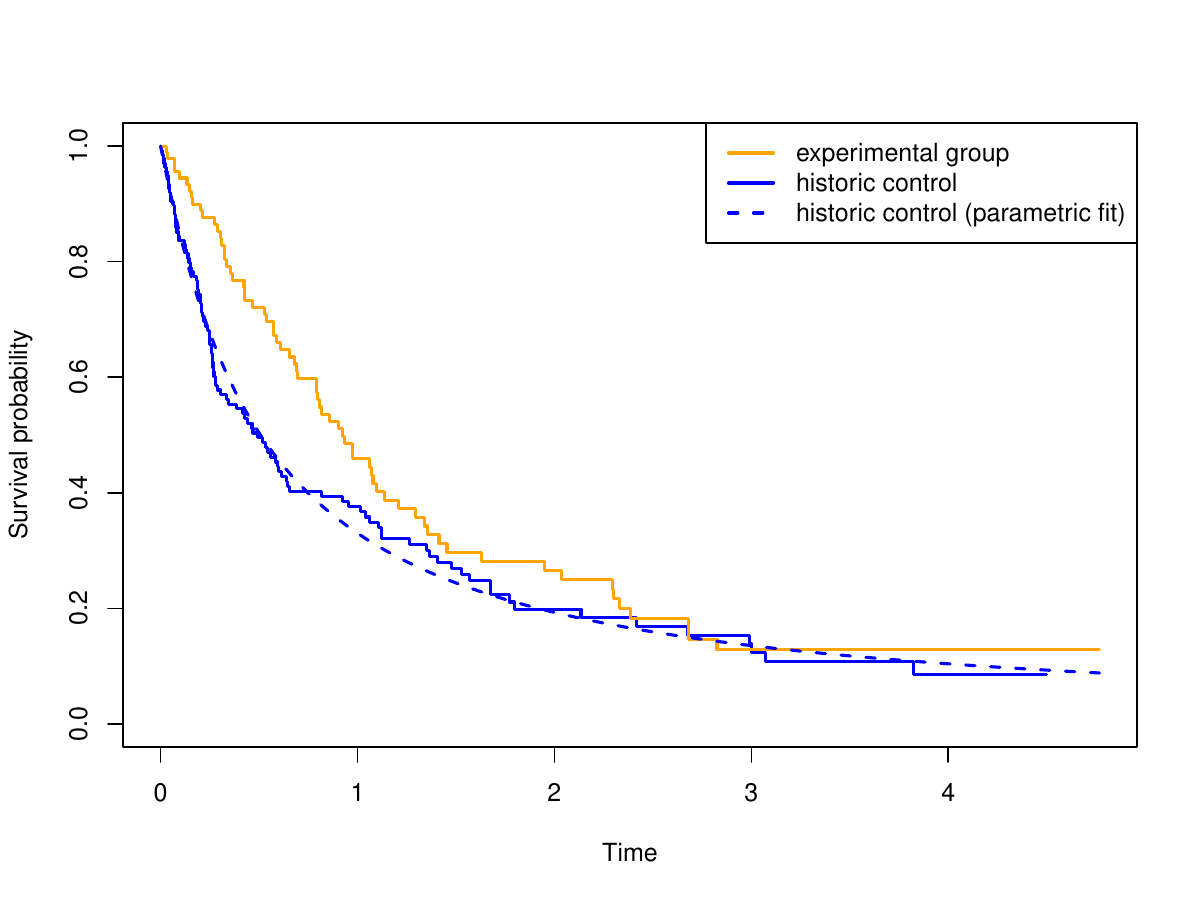"}
	\caption{Kaplan-Meier estimates from simulated data inspired by reconstructed data and the survival curve of a log-logistic distribution fitted to the historic control group data.}
	\label{fig:surv_curves_casestudy}
\end{figure}%
As individual patient data was not available, we reconstructed this data from the Kaplan-Meier curves for overall survival in Figures 2B and 2D of \cite{Fox:2014}. For this purpose we applied the R package \texttt{IPDfromKM} \cite{Liu:2021}. We cut off the data at 4.5 years, as no more events can be observed after that. As shown in \cite{Szurewsky:2025}, the log-logistic distribution fits the data best as it has the lowest AIC among several candidates. Hence, we also applied this distribution to the control data. For our reconstructed data, we obtained the following maximum likelihood estimate of the reference cumulative hazard function:
\begin{equation*}
	\Lambda_{T_A, \text{ref}}(t) = \log \left( 1 + \left(\frac{t}{0.5328} \right)^{1.0760} \right).
\end{equation*}
The variance estimate of our maximum likelihood point estimator $\hat{\boldtheta}_A \coloneqq (1.0760,\, 0.5328)$ is given by
\begin{equation*}
n_A \cdot \mathcal{J}(\hat{\boldtheta}_A) = 
\begin{pmatrix}
	0.9581 & -0.0403\\
	-0.0403 & 0.7242
\end{pmatrix}.
\end{equation*}
Kaplan-Meier curves derived from simulated data based on the data obtained from the reconstruction, as well as the fitted log-logistic survival function, can be found in Figure \ref{fig:surv_curves_casestudy}. As indicated by \eqref{eq:mle_var}, we also require the gradient of the cumulative hazard function w.r.t. the parameter components. For the log-logistic distribution, this is given by
\begin{equation*}
	\nabla_{\boldtheta} \Lambda(\hat{\boldtheta}_A, s) = 
	\begin{pmatrix}
		\left(\frac{s}{\hat{\theta}_{A,2}}\right)^{\hat{\theta}_{A,1}} \cdot \frac{\log\left(\frac{s}{\hat{\theta}_{A,2}}\right)}{\left( 1 + \left(\frac{s}{\hat{\theta}_{A,2}}\right)^{\hat{\theta}_{A,1}} \right)}\\
		-\hat{\theta}_{A,1} \cdot s \cdot \frac{\left(\frac{s}{\hat{\theta}_{A,2}}\right)^{\hat{\theta}_{A,1} - 1}}{\hat{\theta}_{A,2}^2\cdot \left( 1 + \left(\frac{s}{\hat{\theta}_{A,2}}\right)^{\hat{\theta}_{A,1}} \right)}
	\end{pmatrix}.
\end{equation*}
We will now take a closer look at four different test statistics. All of these are based on the compensated counting process $M_{\text{OSLR}}$. However, we are left with two choices when estimating its variances. On the one hand, we can choose between different variance estimators for $M_{\text{OSLR}}$ under $H_{0, \text{OSLR}}$. For the sake of simplicity, we only consider the standard choice of $w=0$ in \eqref{eq:oslr_var} and Wu's suggestion of $w=0.5$. On the other hand, we can also take into account the additional variability that emerges when we want to regard this test as a test of $H_0$ and not only of $H_{0, \text{OSLR}}$. In this case we can add the variance estimate $\hat{V}_2$ as defined in \eqref{eq:mle_var}. For both issues, the options named here can be combined and after adding the variance components, we can divide the compensated counting process by its square root to obtain a $Z$-score and compute one- or two-sided $p$-values. The one-sided $p$-values, for which small values indicate superiority of the experimental group, and the components from which these can be computed can be found in Table \ref{table:case_study}.\\
\begin{table}
	\centering
	\begin{tabular}{l|c|c|c|c|c}
		Testing method & $M_{\text{OSLR}}(4.5)$ & $\hat{V}_1(w, 4.5)$ & $\hat{V}_2(4.5)$ & $Z$ & $p$-value\\
		\hline
		uncorrected OSLR test ($w=0$) & -2.9855 & 1.0492 & --- & -2.9146 & 0.0018\\
		uncorrected OSLR test ($w=0.5$) & -2.9855 & 0.8927 & --- & -3.1598 & 0.0008\\
		corrected OSLR test ($w=0$) & -2.9855 & 1.0492 & 0.7105 & -2.2506 & 0.0122\\
		corrected OSLR test ($w=0.5$) & -2.9855 & 0.8927 & 0.7105 & -2.3579 & 0.0092
	\end{tabular}
	\caption{Description of components of the different test statistics and their results. The total variance estimate is given as the sum of $\hat{V}_1$ and $\hat{V}_2$. The $Z$-score is given by dividing $M_{\text{OSLR}}$ by the square root of the total variance. The resulting $p$-values are one-sided.}
	\label{table:case_study}
\end{table}%
The variance estimate $\hat{V}_1(0.5,\cdot)$ leads to smaller estimates than $\hat{V}_1(0,\cdot)$ which leads to more liberal tests. However, as elaborated in \cite{Danzer:2023}, the resulting test of $H_{0,\text{OSLR}}$ still maintains the significance level for high event rates which is the case in this example. As already lined out in Section \ref{sec:testing}, the $p$-values of the testing procedures that take the estimation uncertainty into account are higher than those of the uncorrected procedures. However, all four procedures would come to the same result if they were used with the standard one-sided significance level of $0.025$.

\section{Simulation study}\label{sec:sim_study}
The aim of this simulation study is to quantify the type I error rate inflation of the one-sample log-rank test when it is interpreted as a test of the null hypothesis $H_0$ instead of $H_{0,\text{OSLR}}$ in settings where the reference curve is given by a parametric estimate from historic control data. At the same time, we want to investigate whether the corrected test, that accounts for the additional variability $V_2$ introduced by the estimation procedure, maintains the nominal type I error rate. Furthermore, we want to examine whether the corrected test has power advantages over the two-sample log-rank test due to the parametric assumption for the historical control cohort.

\subsection{Data generation}
For the historic control arm as well as for the experimental arm, we assume that the patients were recruited uniformly between the start and the end of the accrual period after $a=2$ years. Each patient remains in the trial until the event of interest occurs or the end of the follow-up period of length $f=3$ years that follows the accrual period. Hence, we have $C_x \sim \mathcal{U}[f, a+f]$  for $x \in \{A,B\}$. The recruitment date and hence also the date of administrative censoring are simulated independently from the time of the event of interest. As our aim is to analyse the adherence to the nominal type I error rate, we use the same distribution for $T_A$ and $T_B$. We assume a Weibull distribution for the historical control group with shape parameter $\kappa \in \{0.5, 1, 2\}$ with a fixed one-year survival rate of 50\%. Hence, the cumulative hazard function at time $s$ is given by $\Lambda_{T_A}(s) = -\log(0.5) \cdot s^{\kappa}$. The cumulative hazard function for the experimental group is given by a proportional hazards model via $\Lambda_{T_B}(s) = \omega \cdot\Lambda_{T_A}(s)$. In order to investigate type I error rates as well as power, we considered considered hazard ratios $\omega \in \{1, 0.8\}$. Please note that a Weibull distribution with shape parameter $\kappa=1$ corresponds to an exponential distribution.\\
We used different sample sizes $n_B \in \{25,50,100,200\}$ and allocation ratios $\pi \in \{1, 1/2, 1/4, 1/8, 1/16\}$ to study the impact of these two variables on the different error rates. As it can be seen from \eqref{eq:add_var}, the allocation ratio $\pi$ plays an essential role when it comes to the additional variability introduced by the estimation process.\\
For each combination of those parameters we conduct 100,000 simulation runs. In each run, we use the R package \texttt{flexsurv} to fit a distribution to the data from group $A$ and conduct the standard one-sample log-rank test to compare the data from group $B$ with the resulting reference curve. At the same time, we also compute the corrected test statistics and the two-sample log-rank test statistic comparing the two groups directly. For all scenarios, we use Weibull MLE estimates for the procedures that require parametric estimates. If $\kappa = 1$, the time-to-event endpoint is exponentially distributed. In these scenarios, we will also consider estimates for an exponential distribution. In this way, we can also uncover potential advantages of a more sparse parametric modeling of the refrence survival curve. The Monte Carlo error interval for 100,000 simulation runs for the nominal (one-sided) rate of 0.025 is $[0.0240; 0.0260]$. For the (two-sided) rate of 0.05 it amounts to $[0.0486; 0.0514]$.
\begin{table}[h]
	\centering
	\begin{tabularx}{\textwidth}{l|X}
		Abbreviation  &  Description\\ 
		\hline 
		\textbf{OSLR} & One-sample log-rank test with true (but in practice unknown) reference function.\\
		\textbf{OSLR (Wu)} & As above, but with variance estimation according to \cite{Wu:2014}, i.e. $w=0.5$ in \eqref{eq:oslr_var}.\\ 
		\textbf{Uncorrected OSLR} & One-sample log-rank test with parametrically estimated reference curve, but without correction for additional variability. If $\kappa=1$, we also consider a version with estimation from exponential distribution.\\
		\textbf{Uncorrected OSLR (Wu)} & As above, but with variance estimation according to \cite{Wu:2014}, i.e. $w=0.5$ in \eqref{eq:oslr_var}. If $\kappa=1$, we also consider a version with estimation from exponential distribution.\\
		\textbf{Corrected OSLR} & One-sample log-rank test with parametrically estimated reference curve and corrected variance estimation. If $\kappa=1$, we also consider a version with estimation from exponential distribution.\\
		\textbf{Corrected OSLR (Wu)} & As above, but with variance estimation for the OSLR contribution according to \cite{Wu:2014}, i.e. $w=0.5$ in \eqref{eq:oslr_var}. If $\kappa=1$, we also consider a version with estimation from exponential distribution.\\
		\textbf{TSLR} & Direct comparison of the two groups via two-sample log-rank test. 
	\end{tabularx}
	\caption{Overview of the different testing procedures.}
	\label{table:testing_procedures}
\end{table}%

\subsection{Results}
\subsubsection{Type I error rate}\label{subsec:sim_study_t1e}
In the following results, we will compare empirical type I error rates of the seven testing procedures given in Table \ref{table:testing_procedures}.\\
\begin{figure}[h]
	\centering
	\includegraphics[width=\textwidth]{"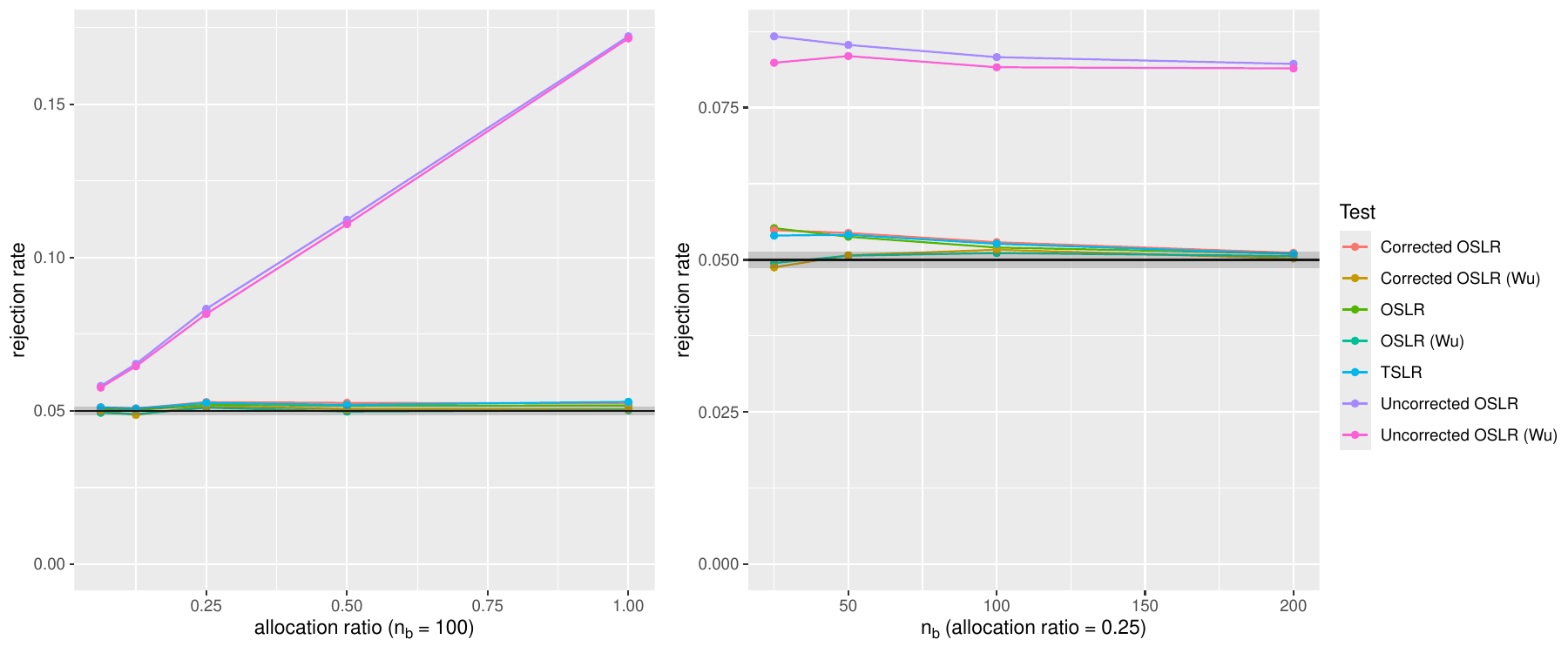"}
	\caption{Empirical two-sided type I error rates for seven different testing procedures with fixed sample size in the experimental group and varying allocation ratios (left) or varying sample size in the experimental group with fixed allocation ratio (right).}
	\label{fig:ts_rates_k1}
\end{figure}%
In Figure \ref{fig:ts_rates_k1} the two-sided empirical type I error rates for all procedures are shown for $\kappa=1$. On the left side, the sample size of the experimental group is fixed at $n_B=100$. Obviously, the two uncorrected procedures that use the parametrically estimated reference curve do not at all comply with the nominal type I error level when regarded as a test of $H_0$. As in \cite{Danzer:2023}, the discrepancy increases drastically as the allocation ratio increases. Within the considered range of allocation ratios, the increase appears to be almost linear in the allocation ratio. All other procedures match to a greater or lesser extent the nominal rate. This includes the two procedures with the additional correction for the application of the estimated reference curve. However, probably due to small sample sizes, not all empirical rates lie within the (grey shaded) Monte Carlo error interval. This is also the case for the standard testing procedure for two-sample survival data, the two-sample log-rank test. On the right hand side of Figure \ref{fig:ts_rates_k1}, the allocation ratio is fixed at $\pi= 0.25$ and the sample size is varied. The type I error rate inflation of the uncorrected procedures is hardly affected by this variation. This resembles the results in \cite{Danzer:2023} and can also be explained by our theory as none of the factors of the additional variance $V_2$ in \eqref{eq:add_var} is affected by the variation of the sample size while the allocation ratio is held constant. However, for the uncorrected as well as for the corrected procedures and the two-sample log-rank test we can observe a slightly larger inflation for small sample sizes. This may be because the asymptotics are not yet taking effect here. As can be seen from the results in the Supplementary Material, there is no remarkable difference in these rates depending on whether an exponential or a Weibull distribution is assumed in the case of $\kappa=1$.\\
\begin{figure}[h]
	\centering
	\includegraphics[width=\textwidth]{"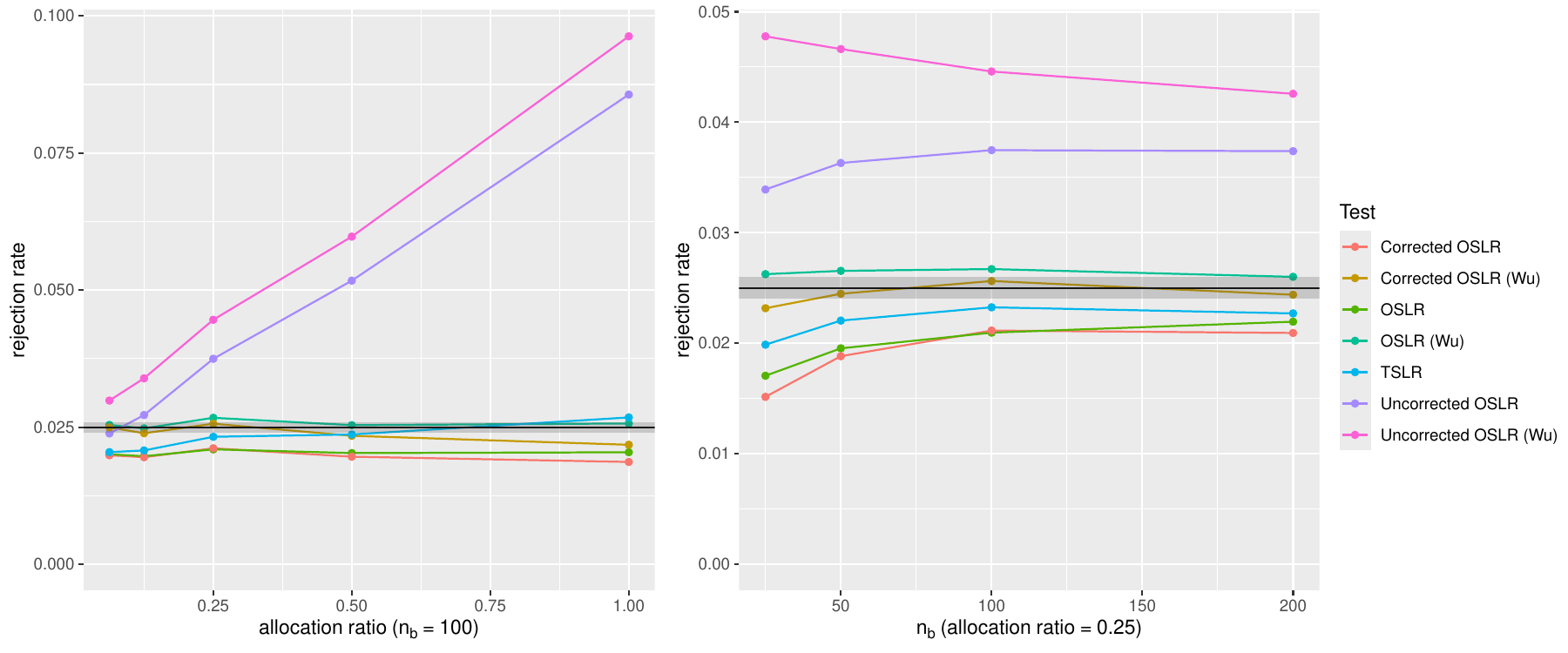"}
	\caption{Empirical one-sided type I error rates for seven different testing procedures with fixed sample size in the experimental group and varying allocation ratios (left) or varying sample size in the experimental group with fixed allocation ratio (right).}
	\label{fig:os_left_rates_k1}
\end{figure}%
Similar observations can be made for the one-sided rates shown in Figure \ref{fig:os_left_rates_k1}. These rates show the rejection rates of $H_0$ in favour of the experimental treatment. The choice of the weight $w$ in \eqref{eq:oslr_var} plays a much more important role here. While the rejection rates for the procedures using $w=0$ and $w=0.5$ were almost indistinguishable in Figure \ref{fig:ts_rates_k1}, this is not the case here anymore. As also found in \cite{Wu:2014, Danzer:2023} the standard variance estimation leads to strongly conservative tests when trying to prove superiority of the new treatment. Choosing $w=0.5$ as in \cite{Wu:2014} yields a noticeably less conservative testing procedure. In particular, the corrected one-sample log-rank test performs quite well for unbalanced sample sizes but the performance deteriorates when the allocation ratio increases and the overall sample size decreases. This development is contrary to that for the standard two-sample test. As already shown in \cite{Kellerer:1983}, the two-sample log-rank test is conservative in unbalanced settings. However, performance improved if the sample size of the control cohort decreases.\\
As with the two-sided rates, it makes no noticeable difference whether the more restrictive assumption of an exponential distribution is made for the case $\kappa=1$. Results for other settings are very similar to those shown here. Corresponding plots can be found in the Supplementary Material.

\subsubsection{Power}
As the results for $\omega=1$, i.e. under the null hypothesis $H_0$, have shown, the corrected versions of the one-sample log-rank test are indeed suitable as a test for this hypothesis. Consequently, we now want to evaluate whether the corrected test achieves efficiency gains by making a parametric assumption about one of the two groups instead of using a completely non-parametric method.\\
\begin{figure}[h]
	\centering
	\includegraphics[width=\textwidth]{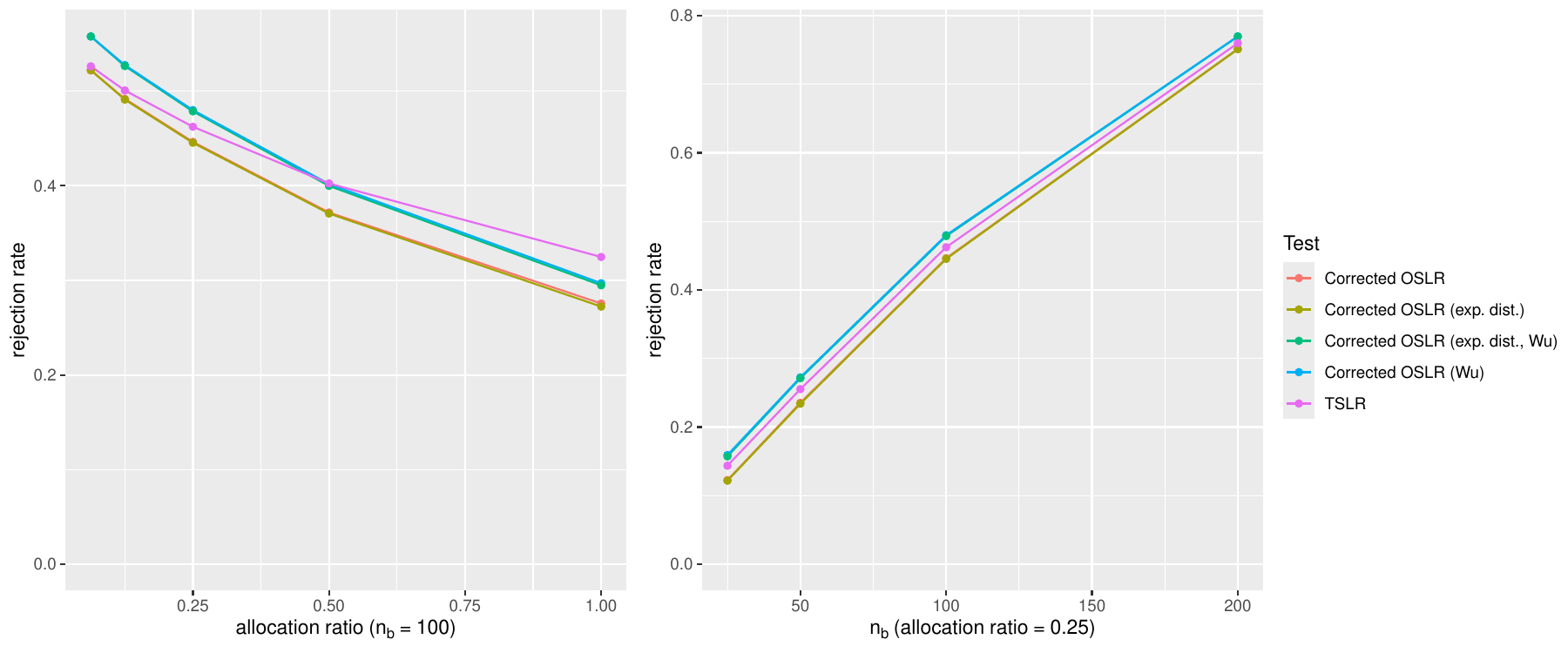}
	\caption{Empirical power for five procedures that showed acceptable adherence to the nominal type I error rate with fixed sample size in the experimental group and varying allocation ratios (left) or varying sample size in the experimental group with fixed allocation ratio (right).}
	\label{fig:power_k1}
\end{figure}%
For the hazard ratio $\omega=0.8$ between the two groups, Figure \ref{fig:power_k1} shows the power of the corrected procedures and the two-sample log-rank test. In this figure, we consider the case $\kappa=1$ so that we can also meaningfully compare tests in which Weibull MLEs are used with those in which exponential MLEs are used. Both suffice the standard assumption discussed in Section \ref{sec:par_est} and hence are consistent and asymptotically normal. However, as the shape parameter has to be estimated for a Weibull distribution, the exponential estimate might be more efficient and thus yield a higher power of the associated test. However, our results reveal that these expectations are not being met. The corresponding lines in Figure \ref{fig:power_k1} overlap almost completely. This is not only the case for the presented results with $n_B = 100$ and $\pi = 0.25$, but also for the other scenarios which are presented in the Supplementary Material.\\
Compared to the two-sample log-rank test, there are no clear advantages resulting from the parametric assumption. If Wu's variance estimation is used, there are advantages of a few percentage points for small allocation ratios. However, for allocation ratios close to 1, this effect is reversed. These observations are consistent with the conservativeness of the two-sample log-rank test for small allocation ratios observed in the previous section and described in \cite{Kellerer:1983}, and with the conservativeness of the corrected one-sample log-rank test for allocation ratios close to 1, which was also observed above.\\
Results for other settings support these findings an can also be found in the Supplementary Material.

\section{Discussion}\label{sec:discussion}
As we have demonstrated here, a considerable inflation of the type I error rate occurs if the one-sample log-rank test is to be interpreted as a test of the practically highly relevant null hypothesis $H_0$. This is the case as the calculation of the test statistic does not take the estimation uncertainty of the reference curve into account. With the analytical derivation and a simulation study, we were able to demonstrate that we can decompose the variance of the test statistic under the hypothesis $H_0$ into two parts, one of which corresponds to the well-known variance estimator for the one-sample log-rank test and the other addresses the uncertainty in the reference curve. As can be seen from formula \eqref{eq:mle_var}, the allocation ratio between the historical control cohort and the experimental cohort is of fundamental importance. While the one-sample log-rank test in its current form can still play an important role as a helpful tool, particularly in the early stages of treatment development, especially in rare diseases or in paediatrics, we would like to draw attention to the fact that there are ways of taking a more critical statistical view of the role of the reference curve. For example, the test statistics developed here could be considered as an additional sensitivity analysis to the one-sample log-rank test.\\
Of course, it can be argued that the resulting comparison is ultimately a two-group comparison between the historical cohort and the study cohort, which could also be achieved with a two-sample log-rank test. By making an additional parametric assumption about the distribution of the endpoint in the historical control group, it seems conceivable at first glance that the resulting tests are more efficient and allow for gains in terms of power in comparison with the two-sample log-rank test, which does not require such assumptions. However, our simulation study revealed that no systematic gains can be achieved in a proportional hazards scenario. In comparison with the two-sample log-rank test there is a slight advantage in cases where the historical control cohort is relatively large in comparison with the experimental group. For groups of similar size, on the other hand, there is a slight disadvantage. These findings are consistent with the observed conservativeness of the respective methods described in Section \ref{subsec:sim_study_t1e}. Analogously to the uncorrected one-sample procedures, the choice of the estimator for the variance component $V_1$ plays a crucial role for the corrected tests in terms of power. We refer to \cite{Danzer:2023} for a more extensive discussion of this issue.\\
In summary, based on our simulation results, it can be said that conducting the two-sample log-rank test is preferable if the necessary data is available and a test of the hypothesis $H_0$ is desired instead of a comparison with a fixed reference curve. However, it should be noted that our simulations are limited to proportional hazards scenarios, in which the two-sample log-rank test is known to be very efficient. If the one-sample log-rank test shall still be applied for the primary analysis, our proposed method could be still be used as an additional sensitivity analysis or as a tool to determine how much of the total variance of the test statistics can be attributed to the uncertainty inherent in the estimate of the reference curve.\\  
Going beyond the standard one-sample log-rank test, the techniques presented and evaluated here are in general also applicable to variations of the one-sample log-rank test. This includes weighted tests \cite{Chu:2020}, and resulting combination tests \cite{Szurewsky:2025} as well as multi-stage procedures \cite{Kwak:2014, Schmidt:2018}.\\
Concerning the estimation process in the historical control cohort, we assume that the parametric model is correctly specified. If this is not the case, the estimator will converge against a 'pseudo-true' value and the variance can instead be estimated by a sandwich estimator as demonstrated in \cite{White:1982}. Use of such methods could improve the robustness of the presented method. In practice, however, the determination of the reference curve might include a selection procedure to decide which candidate among several parametric families fits the data best. As already mentioned, this selection can be accomplished using the AIC. Such a selection has also been made in \cite{Szurewsky:2025} for the data example that is also used by us in Section \ref{sec:application}. However, this selection procedure can introduce additional variability that has not been accounted for in our approach.\\
Further improvements concerning the adherence to the nominal type I error level could be achieved by resampling procedures. In order to account for the variability introduced by the estimation process, we applied the Delta Method which yields asymptotic normality of the hazards accumulated by each patient in the experimental group. However, for this particular transformation a normal distribution might not yet be met for small sample sizes. The actual distribution of the transformed parameters could be mimicked by a resampling approach of the historical control cohort.

\section*{Acknowledgments}
Funded by the Deutsche Forschungsgemeinschaft (DFG, German Research Foundation) - 413730122. Parts of the calculations for this publication were performed on the HPC cluster PALMA II of the University of Münster, subsidised by the DFG (INST 211/667-1).

\section*{Data and code availability}
The code and results of the simulation study, simulated data based on the case study, and the code we used to reconstruct and evaluate the case study data are available in a public repository that can be accessed at \url{https://github.com/moedancer/MLEVar_OSLR}.

\putbib
\end{bibunit}

\newpage 

\appendix
\renewcommand{\thesection}{\Alph{section}}
\setcounter{figure}{0}
\setcounter{table}{0}
\renewcommand{\figurename}{Supplementary Figure}
\renewcommand{\tablename}{Supplementary Table}
\renewcommand{\thetable}{S\arabic{table}}
\renewcommand{\thefigure}{S\arabic{figure}}

\begin{bibunit}
	
\begin{center}
	\huge
	Supplementary Material for\\
	'Correcting for sampling variability in maximum likelihood-based one-sample log-rank tests'
\end{center}

\vspace{12pt}

\section{Derivation of the asymptotic distribution}

\begin{lemma}\label{lemma:convergence_in_distribution_independence}
Let $(X_n)_{n\in\mathbb{N}}$ and $(Y_n)_{n\in\mathbb{N}}$ two sequences of $\mathbb{R}$-valued random variables where $X_n$ and $Y_n$ are independent for any $n \in \mathbb{N}$. If $X_n \overset{\mathcal{D}}{\to} X$ and $Y_n \overset{\mathcal{D}}{\to} Y$ as $n \to \infty$,  then $X_n + Y_n \overset{\mathcal{D}}{\to} X^{\star}+Y^{\star}$ as $n\to\infty$, where $X^{\star}$ and $Y^{\star}$ are copies of $X$ resp. $Y$ s.t. $X^{\star}$ and $Y^{\star}$ are independent.
\end{lemma}
\begin{proof}
As $X_n$ and $Y_n$ are independent for any $n \in \mathcal{N}$, the characteristic function of $X_n + Y_n$ is given by
\begin{equation*}
\Phi_{X_n + Y_n}(t) = \Phi_{X_n}(t) \cdot \Phi_{Y_n}(t) \quad \forall t\in\mathbb{R}.
\end{equation*}
By Lévy's continuity theorem, convergence in distribution is equivalent to pointwise convergence of the corresponding characteristic functions $\Phi$.  Hence, by the convergence of the two sequences we have for all $t \in \mathbb{R}$
\begin{align*}
\lim_{n\to\infty} \Phi_{X_n + Y_n}(t) &=\lim_{n\to\infty} \Phi_{X_n}(t) \cdot \Phi_{Y_n}(t)\\
&=\Phi_{X}(t) \cdot \Phi_{Y}(t)\\
&=\Phi_{X^{\star}}(t) \cdot \Phi_{Y^{\star}}(t)\\
&=\Phi_{X^{\star}+Y^{\star}}(t).
\end{align*}
Therefore, $X_n + Y_n \overset{\mathcal{D}}{\to} X^{\star}+Y^{\star}$ as $n\to\infty$.
\end{proof}

\begin{lemma}\label{lemma:uniform_convergence}
	Let $(f_n)_{n \in \mathbb{N}}$ be a sequence of non-decreasing and continuous functions from a bounded interval $[a,b]$ to $\mathbb{R}$. This sequence converges pointwise to a non-decreasing continuous function $f$. Then, the convergence is also uniform.
\end{lemma}
\begin{proof}
	As $f$ is a continuous function on a bounded interval, it is also uniformly continuous. Hence, for any $\varepsilon > 0$, we can find a $\delta > 0$ such that $|u-v| < \delta$ implies $|f(u) - f(v)| < \varepsilon/2$ for all $u,v \in [a,b]$. In particular, we can find a finite subdivision $a = t_0 < t_1 < \dots < t_m = b$ of the interval such that
	\begin{equation}\label{eq:bound_subdivision_f}
		f(t_i) -  \varepsilon/2 \leq f(t) \leq f(t_{i-1}) + \varepsilon/2
	\end{equation} for any $t \in [t_{i-1}, t_i]$. We consider the points of the subdivision and we can find some $N$ large enough such that $|f_n(t_i) - f(t_i)|$ for all $n \geq N$ and all subdivision points $t_i$.\\
	For an arbitrary $t \in [a,b]$, let $t_{i-1}$ and $t_i$ be the subdivision points between which $t$ lies. From the monotonicity of $f_n$ and the distance between $f_n$ and $f$ on the subdivision points, it follows that 
	\begin{equation*}
		f(t_{i-1}) - \frac{\varepsilon}{2} \leq f_n(t_{i-1}) \leq f_n(t) \leq f_n(t_i) \leq f(t_i) + \frac{\varepsilon}{2}.
	\end{equation*}
	Application of \eqref{eq:bound_subdivision_f} yields
	\begin{equation*}
		f(t) - 2\frac{\varepsilon}{2} \leq f_n(t) \leq f(t) + 2\frac{\varepsilon}{2}, 
	\end{equation*}
	i.e. $|f_n(t) - f(t)| \leq \varepsilon$ which finishes the proof as $n$ was chosen s.t. it exceeded $N$ and $t$ and $\varepsilon$ were chosen arbitrarily.
\end{proof}

\begin{lemma}\label{lemma:empirical_processes}
	Let $(\hat{\boldtheta}_n)_{n \in \mathbb{N}}$ be a sequence of $\Theta$-valued random variables s.t. $\hat{\boldtheta}_n \overset{\mathbb{P}}{\to} \boldtheta_0$ for some constant $\boldtheta$ and let $(X_n)_{n \in \mathbb{N}}$ be a sequence of i.i.d. $\mathbb{R}$-valued random variables which is also independent of $\hat{\boldtheta}_m$ for any $m \in \mathbb{N}$. Let also $G\colon \Theta \times \mathbb{R}_+ \to \mathbb{R}_+$ be a function that fulfills the following properties
	\begin{enumerate}[label = (\roman*)]
		\item \label{item:cont} $G$ is continuous,
		\item \label{item:monotonicity} $G(\boldtheta,\cdot)$ is a monotone function for each $\boldtheta \in \Theta$,
		\item \label{item:boundedness} $G(\boldtheta, s) \leq C$ for all $\boldtheta \in \Theta$ and $s \in [0,\tmax]$ for some $0<C<\infty$.
	\end{enumerate}
	If $\Theta$ contains a neighbourhood $U$ of $\boldtheta_0$, then it holds
	\begin{enumerate}[label = (\arabic*)]
		\item \label{item:lln} a weak law of large numbers, i.e.
		\begin{equation}
			\frac{1}{n}\sum_{i=1}^n G(\hat{\boldtheta}_n, X_i) \overset{\mathbb{P}}{\to} \mathbb{E}[G(\boldtheta_0, X_1)].
		\end{equation}
		as $n \to \infty$ and
		\item \label{item:asymp_equicont} the asymptotic continuity property
		\begin{align*}
			&n^{-\frac{1}{2}} \left( \sum_{i=1}^n (G(\hat{\boldtheta}_n, X_i) - \mathbb{E}[G(\hat{\boldtheta}_n, X_1)]) -\sum_{i=1}^n (G(\boldtheta_0, X_i) - \mathbb{E}[G(\boldtheta_0, X_1)])\right)\\
			&\overset {\mathbb{P}}{\to}0.
		\end{align*}
		as $n \to \infty$.
	\end{enumerate} 
\end{lemma}

\begin{proof}
	These statements can be easily proven by application of empirical process theory to the empirical process $\mathbb{G}_n$ that is defined by its evaluation at $\boldtheta \in \Theta$ via
	\begin{equation*}
		\mathbb{G}_n (\boldtheta) \coloneqq	n^{-1/2} \sum_{i=1}^n (G(\boldtheta, X_i) - \mathbb{E}[G(\boldtheta, X_1)])
	\end{equation*}
	for each $\boldtheta \in \Theta$.\\
	The two statements will be consequences of the fact that the class of functions $G$ fulfilling the assumptions above is Donsker. In this context, it should be mentioned that subclasses of Donsker classes are also Donsker classes and that Donsker classes are also Glivenko-Cantelli classes (see e.g. \cite{vanderVaart:1996, vanderVaart:1998, Pollard:1990}).\\
	According to Theorem 2.7.9 of \cite{vanderVaart:1998}, the class of uniformly bounded, monotone functions is Donsker. The $\Theta$-indexed subclass examined here is a subclass and hence also a Donsker and a Glivenko-Cantelli class. Then, statement \ref{item:asymp_equicont} follows from the resulting property of asymptotic equicontinuity (see Section 2.1.2 of \cite{vanderVaart:1996} and Equation (10.4) of \cite{Pollard:1990}). For this, we also require the uniform convergence proven in Lemma \ref{lemma:uniform_convergence} that guarantees the requirement stated before Equation (10.4) of \cite{Pollard:1990}.\\
	For statement \ref{item:lln} the triangle inequality implies
	\begin{align*}
		&\mathbb{P}\left[ \left| \frac{1}{n}\sum_{i=1}^n G(\hat{\boldtheta}_n, X_i)- \mathbb{E}[G(\boldtheta_0, X_1)] \right| > \varepsilon \right]\\
		\leq& \mathbb{P}\left[ \left| \frac{1}{n} \sum_{i=1}^n \left(G(\hat{\boldtheta}_n, X_i) - \mathbb{E}[G(\hat{\boldtheta}_n, X_1)|\hat{\boldtheta}_n]\right) \right| > \frac{\varepsilon}{2} \right]\\
		&+\mathbb{P}\left[ \left| \mathbb{E}[G(\hat{\boldtheta}_n, X_1)|\hat{\boldtheta}_n] - \mathbb{E}[G(\boldtheta_0, X_1)] \right| > \frac{\varepsilon}{2} \right]
	\end{align*}
	for arbitrary $\varepsilon > 0$. We have to show that this quantity is smaller than any $\delta>0$ for $n \in \mathbb{N}$ large enough.\\
	For the second summand we want to establish convergence in probability of the sequence $(\mathbb{E}[G(\hat{\boldtheta}_n, X)|\hat{\boldtheta}_n])_{n\in\mathbb{N}}$ to $\mathbb{E}[G(\boldtheta, X)]$. By the continuous mapping theorem (see \cite{vanderVaart:1998}, Theorem 2.3), we need continuity of the mapping $\tilde{\boldtheta}\mapsto \mathbb{E}[G(\tilde{\boldtheta}, X)]$ at $\boldtheta_0$ for that. The continuity of $\tilde{\boldtheta}\mapsto G(\tilde{\boldtheta}, x)$ at $\boldtheta_0$ for any $x \in \mathbb{R}_+$ and the uniform bound allows application of the dominated convergence theorem to prove continuity. Hence, we can find $n^{(1)}$ large enough, s.t. the second summand summand is smaller than $\delta/2$ for any $n > n^{(1)}$.\\
	For the second summand, we note that
	\begin{align*}
		&\mathbb{P}\left[ \left| \frac{1}{n} \sum_{i=1}^n \left(G(\hat{\boldtheta}_n, X_i) - \mathbb{E}[G(\hat{\boldtheta}_n, X_1)|\hat{\boldtheta}_n]\right) \right| > \frac{\varepsilon}{2} \right]\\
		\leq&\mathbb{P} \left[ \sup_{\tilde{\boldtheta} \in U} \left| \frac{1}{n} \sum_{i=1}^n \left(G(\tilde{\boldtheta}, X_i) - \mathbb{E}[G(\tilde{\boldtheta}, X_1)]\right) \right| > \frac{\varepsilon}{2} \right]\\
		&+\mathbb{P}[\hat{\boldtheta}_n \notin U]\\
		\leq&\mathbb{P} \left[ \sup_{\tilde{\boldtheta} \in \Theta} \left| \frac{1}{n} \sum_{i=1}^n \left(G(\tilde{\boldtheta}, X_i) - \mathbb{E}[G(\tilde{\boldtheta}, X_1)]\right) \right| > \frac{\varepsilon}{2} \right]\\
		&+\mathbb{P}[\hat{\boldtheta}_n \notin U]
	\end{align*}
	As the class of functions examined here is Glivenko-Cantelli, the convergence $1/n \sum_{i=1}^n G(\boldtheta, X_i) \overset{\mathbb{P}}{\to} \mathbb{E}[G(\boldtheta, X_1)]$ is uniform over $\boldtheta \in \Theta$. Hence, we can find $n^{(2)}$ large enough, s.t. the corresponding probability is smaller than $\delta/4$ for any $n > n^{(2)}$. By the convergence in probability of $(\hat{\boldtheta}_n)_{n \in \mathbb{N}}$, there is also some $n^{(3)}$ large enough, s.t. the last summand here is smaller than $\delta/4$ for any $n > n^{(3)}$.\\
	Putting this together, we find that 
	\begin{equation*}
		\mathbb{P}\left[ \left| \frac{1}{n}\sum_{i=1}^n G(\hat{\boldtheta}_n, X_i)-\mathbb{E}[G(\boldtheta, X_1)] \right| > \varepsilon \right] < \delta \quad \forall n > \max(n^{(1)}, n^{(2)}, n^{(3)})
	\end{equation*}
	which concludes the proof of statement \ref{item:lln}.
\end{proof}

\begin{theorem}\label{thm:main_thm}
	We assume that the parametric family of distributions fulfills assumptions \ref{item:continuity}-\ref{item:exchange_expectation_differentiation}. Let $\hat{\boldtheta}_A$ be the maximum likelihood estimator of $\boldtheta_A$, which is estimated from the $n_A$ observations from group $A$ and $0 < t \leq \tmax$. If $T_A$ and $T_B$ have the same distribution and under the assumption of independent right-censoring in both groups, we have
	\begin{equation}
	n_B^{-1/2} \left(N_B(t) - \sum_{i \in \mathcal{N}_B} \Lambda(\hat{\boldtheta}_A, t \wedge X_{B,i}) \right) \overset{\mathcal{D}}{\to} \mathcal{N}(0, V_1(t) + V_2(t))
	\end{equation}
	where
	\begin{align}
	V_1(t)&\coloneqq \mathbb{P}[T_B \leq C_B \wedge t]=\mathbb{E}[\Lambda(\boldtheta_A, t \wedge X_B)],\\
	V_2(t)&\coloneqq \pi \cdot\mathbb{E}[\nabla_{\boldtheta} \Lambda(\boldtheta_A, t \wedge X_B)]^T \mathcal{I}(\boldtheta_A)^{-1} \mathbb{E}[\nabla_{\boldtheta} \Lambda(\boldtheta_A, t \wedge X_B)]
	\end{align}
as $n_A + n_B = n \to \infty$ with $n_B/n_A \to \pi > 0$.
	\end{theorem}
\begin{proof}
First, we split up the above term as follows
\begin{align}
&n_B^{-1/2} \left( N_B(t) - \sum_{i \in \mathcal{N}_B} \Lambda(\hat{\boldtheta}_A, t \wedge X_{B,i}) \right)\\
=&n_B^{-1/2} \left( N_B(t) - \sum_{i \in \mathcal{N}_B} \Lambda(\boldtheta_A, t \wedge X_{B,i}) \right)  \\
 &+ n_B^{-1/2} \left(\sum_{i \in \mathcal{N}_B} \left(\Lambda(\boldtheta_A, t \wedge X_{B,i})  - \mathbb{E}[\Lambda(\boldtheta_A, t \wedge X_B)] \right) \right) \\
 &- n_B^{-1/2} \left(\sum_{i \in \mathcal{N}_B} \left(\Lambda(\hat{\boldtheta}_A, t \wedge X_{B,i})  - \mathbb{E}[\Lambda(\hat{\boldtheta}_A, t \wedge X_B) | \hat{\boldtheta}_A] \right) \right) \\
 &- n_B^{1/2} \left( \mathbb{E}[\Lambda(\hat{\boldtheta}_A, t \wedge X_B) | \hat{\boldtheta}_A] - \mathbb{E}[\Lambda(\boldtheta_A, t \wedge X_B)] \right)
\end{align}
In the following, we will look at the first summand, the last summand and the two middle summands separately.\\
\underline{\textit{First summand:}}\\
If $T_A$ and $T_B$ have the same distribution, we can replace $\boldtheta_A$ by $\boldtheta_B$ s.t. the resulting process
\begin{equation}
M(t) \coloneqq n_B^{-1/2} \left( N_B(t) - \sum_{i \in \mathcal{N}_B} \Lambda(\boldtheta_A, t \wedge X_{B,i}) \right)
\end{equation}
is a martingale w.r.t. the filtration generated by the observations in group $B$. From Rebolledo's martingale central limit theorem (see e.g. \cite{Andersen:1993}, Theorem II.5.1) we know that the limiting process as $n_B \to \infty$ is a continuous Gaussian martingale with variance function $V_1(t)$ as defined in the statement above. In particular, for a fixed  $t \geq 0$ we have
\begin{equation}
n_B^{-1/2} \left( N_B(t) - \sum_{i \in \mathcal{N}_B} \Lambda(\boldtheta_A, t \wedge X_{B,i}) \right) \overset{\mathcal{D}}{\to} \mathcal{N}(0,V_1(t)).
\end{equation}
as $n_B \to \infty$ and hence also as $n \to \infty$.\\
\underline{\textit{Second and third summand:}}\\
We apply part \ref{item:asymp_equicont} of Lemma \ref{lemma:empirical_processes} to prove that the sum of the second and the third summand converges to 0 in probability. Without loss of generality, we can assume that assumption \ref{item:boundedness} of Lemma \ref{lemma:empirical_processes} is fulfilled. Because otherwise, by local compactness of $\Theta$, there is a compact neighbourhood $K_{\boldtheta_A}$ of $\boldtheta_A$. By its continuity, $\Lambda$ is bounded on $K_{\boldtheta_A} \times [0,\tmax]$. Hence, if necessary, we can restrict $\Theta$ to $K_{\boldtheta_A}$. Also, by consistency of $\hat{\boldtheta}_A$, we know that the probability, that $\hat{\boldtheta}_A$ will lie outside of $K_{\boldtheta_A}$ converges to $0$ as $n \to \infty$.\\
\underline{\textit{Last summand:}}\\
Regarding the last summand, we want to apply the multivariate delta method to $\hat{\boldtheta}_A$ together with the function
\begin{equation}
\boldtheta \mapsto \mathbb{E}[\Lambda(\boldtheta, t \wedge X_B)].
\end{equation}
By assumptions \ref{item:differentiability} and \ref{item:exchange_expectation_differentiation} we obtain continuous differentiability of this function at $\boldtheta_A$ and we can exchange differentiation and integration by application of the dominated convergence theorem, i.e.
\begin{equation}
\nabla_{\boldtheta} \mathbb{E}[\Lambda(\boldtheta_A, t \wedge X_{B,i})] = \mathbb{E}[\nabla_{\boldtheta} \Lambda(\boldtheta_A, t \wedge X_{B,i})].
\end{equation}
Together with the asymptotic normality of $\hat{\boldtheta}_A$, we can apply the multivariate delta method as e.g. given in \cite{vanderVaart:1998}, Theorem 3.1. This yields
\begin{align}
&n_A^{1/2} \left( \mathbb{E}[\Lambda(\hat{\boldtheta}_A, t \wedge X_B) | \hat{\boldtheta}_A] - \mathbb{E}_{X_B}[\Lambda(\boldtheta_A, t \wedge X_B)] \right)\\
\overset{\mathcal{D}}{\to} &\mathcal{N}(0, \mathbb{E}[\nabla_{\boldtheta} \Lambda(\boldtheta_A, t \wedge X_B)]^T \mathcal{I}(\boldtheta_A)^{-1} \mathbb{E}[\nabla_{\boldtheta} \Lambda(\boldtheta_A, t \wedge X_B)])
\end{align}
and hence
\begin{align}
n_B^{1/2} \left( \mathbb{E}[\Lambda(\hat{\boldtheta}_A, t \wedge X_B) | \hat{\boldtheta}_A] - \mathbb{E}_{X_B}[\Lambda(\boldtheta_A, t \wedge X_B)] \right) \overset{\mathcal{D}}{\to} \mathcal{N}(0, V_2(t))
\end{align}
as $n_B=\pi n_A$.\\
\underline{\textit{Conclusion:}}
As the first summand only depends from observations in group $B$ and the last summand only depends from $\hat{\boldtheta}_A$ which is computed from observations in group $A$, they are independent. Hence it follows from Lemma \ref{lemma:convergence_in_distribution_independence} that 
\begin{align}
&n_B^{-1/2} \left( N_B(t) - \sum_{i \in \mathcal{N}_B} \Lambda(\boldtheta_A, t \wedge X_{B,i}) \right) \\
&\qquad - n_B^{1/2} \left( \mathbb{E}_{X_B}[\Lambda(\hat{\boldtheta}_A, t \wedge X_B)] - \mathbb{E}_{X_B}[\Lambda(\boldtheta_A, t \wedge X_B)] \right)\\
\overset{\mathcal{D}}{\to}&\mathcal{N}(0, V_1(t) + V_2(t)).
\end{align}
As the sum of second and the third summand converges in probability to 0, we can conclude.
\end{proof}

\begin{corollary}\label{cor:var_est_pv}
	Let $\hat{\theta}_A$ be the maximum likelihood estimator of $\theta_A$, which is estimated from the $n_A$ observations from group $A$. If $T_A$ and $T_B$ have the same distribution, we have
	\begin{equation}
		\frac{1}{n_B}\sum_{i\in \mathcal{N}_B}\Lambda(\hat{\boldtheta}_A, t\wedge X_{B,i}) \overset{\mathbb{P}}{\to} V_1(t).
	\end{equation}
	as $n_A + n_B = n \to \infty$ with $n_B/n_A \to \pi > 0$ where $V_1(t)\coloneqq\mathbb{P}[T_B \leq C_B \wedge t]=\mathbb{E}[\Lambda(\theta_A, t \wedge X_B)]$.
\end{corollary}

\begin{proof}
	This is a consequence of part \ref{item:lln} of Lemma \ref{lemma:empirical_processes} because without loss of generality, we can assume that assumption \ref{item:boundedness} of Lemma \ref{lemma:empirical_processes} is fulfilled. The reasoning is the same as in the proof of Theorem \ref{thm:main_thm}.
\end{proof}

\begin{theorem}\label{thm:var_est}
Under the set of assumptions made in Section VI.1.2 of \cite{Andersen:1993} and our additional assumptions \ref{item:continuity} - \ref{item:exchange_expectation_differentiation}, the overall variance
$V_1(t) + V_2(t)$ at some time $0 < t \leq \tmax$ can be consistently estimated by 
\begin{align*}
	&\overbrace{\frac{1}{n_B}\cdot \left(w \cdot N_B(t) +(1-w) \cdot \sum_{i \in \mathcal{N}_B} \Lambda(\hat{\boldtheta}_A, t \wedge X_{B,i}) \right)}^{\eqqcolon \hat{V}_1(w,t)}\\
	+&\underbrace{\pi \cdot \left(\frac{1}{n_B} \sum_{i=1}^{n_B} \nabla_{\boldtheta} \Lambda(\hat{\boldtheta}_A, t \wedge X_{B,i}) \right)^T \mathcal{J}(\hat{\boldtheta}_A)^+ \left( \frac{1}{n_B} \sum_{i=1}^{n_B} \nabla_{\boldtheta} \Lambda(\hat{\boldtheta}_A, t \wedge X_{B,i}) \right)}_{\eqqcolon \hat{V}_2(t)}\\ 
\end{align*}
for any $w \in [0,1]$ as $n_A + n_B = n \to \infty$ with $n_B/n_A \to \pi > 0$.
\end{theorem}
\begin{proof}
As convergence in probability is conserved under summation, we can consider $V_1(t)$ and $V_2(t)$ separately.\\
\underline{$V_2(t)$:}\\
As stated in Theorem VI.1.2 of \cite{Andersen:1993}, $\mathcal{I}(\boldtheta_A)$ can be consistently estimated by its empirical counterpart $\mathcal{J}(\hat{\boldtheta}_A)$. From \cite{Puri:1984}, we know that invertibility of $\mathcal{I}(\boldtheta_A)$ also yields
\begin{equation*}
	\mathcal{J}(\hat{\boldtheta}_A)^+ \overset{\mathbb{P}}{\to} \mathcal{I}(\boldtheta_A)^{-1}
\end{equation*}
as $n \to \infty$ where $M^+$ denotes the generalized Moore-Penrose inverse of a matrix $M$.\\
Using our assumptions \ref{item:differentiability} and \ref{item:exchange_expectation_differentiation}, we can apply the Continuous Mapping Theorem (see e.g. Theorem 2.3 in \cite{vanderVaart:1998}) and the dominated convergence theorem to state
\begin{equation*}
	\frac{1}{n_B} \sum_{i=1}^{n_B} \nabla_{\boldtheta} \Lambda(\hat{\boldtheta}_A, t \wedge X_{B,i}) \overset{\mathbb{P}}{\to} \mathbb{E}[\nabla_{\boldtheta} \Lambda(\boldtheta_A, t \wedge X_B)]
\end{equation*} 
With Slutsky's theorem (see e.g. Lemma 2.8 in \cite{vanderVaart:1998}) we can conclude this part of the proof.\\
\underline{$V_1(t)$:}\\
As lined out in \cite{Danzer:2023}, the variance component $V_1(t)$ can be consistently estimated via the quadratic variation and the predictable variation process, i.e.
\begin{equation}
	\frac{1}{n_B} N_B(t) \overset{\mathbb{P}}{\to} V_1(t) \quad \text{and} \quad \frac{1}{n_B} \sum_{i=1}^{n_B} \Lambda(\boldtheta_A, t \wedge X_{B,i}) \overset{\mathbb{P}}{\to} V_1(t).
\end{equation}
Following Corollary \ref{cor:var_est_pv}, we can replace the true parameter by the estimate $\hat{\boldtheta}_A$ an combine the two convexly to obtain the convergence
\begin{equation}
\frac{1}{n_B}\cdot \left(w \cdot N_B(t) +(1-w) \cdot \sum_{i \in \mathcal{N}_B} \Lambda(\hat{\boldtheta}_A, t \wedge X_{B,i}) \right) \overset{\mathbb{P}}{\to} V_1(t)
\end{equation} 
for any $w \in [0,1]$. Hence, for any such $w$, we have a consistent estimator of $V_1(t)$ under the null hypothesis. Please note that theoretically, a weighting factor $w \notin [0,1]$ would also be allowed. However, we do not recommend choosing a non-convex combination of the two consistent estimators of $V_1(t)$.
\end{proof}

\begin{corollary}\label{cor:test_stat}
	Under the assumptions of Theorem \ref{thm:main_thm} and \ref{thm:var_est}, the counting process that is compensated by a parametrically estimated cumulative hazard function can be standardised by the estimated variance to obtain an asymptotically standard normally distributed test statistics, i.e.
	\begin{equation}
		\frac{n_B^{-\frac{1}{2}}\left(N_B(t) - \sum_{i \in \mathcal{N}_B} \Lambda(\hat{\boldtheta}_A, t \wedge X_{B,i})\right)}{\sqrt{\hat{V}_1(w,t) + \hat{V}_2(t)}} \overset{\mathcal{D}}{\to} \mathcal{N}(0,1).
	\end{equation}
	for any $w \in [0,1]$.
\end{corollary}
\begin{proof}
	This result is a consequence of Slutsky's theorem in combination with Theorem \ref{thm:main_thm} and \ref{thm:var_est}.
\end{proof}

\newpage

\section{Simulation results}

\subsection{Type I error rate}

\begin{figure}[h!]
	\centering
	\begin{tabular}{c|c|c}
		$n_B$ & two-sided rates & one-sided rates \\
		\hline
		25 &\includegraphics[width=.44\textwidth]{"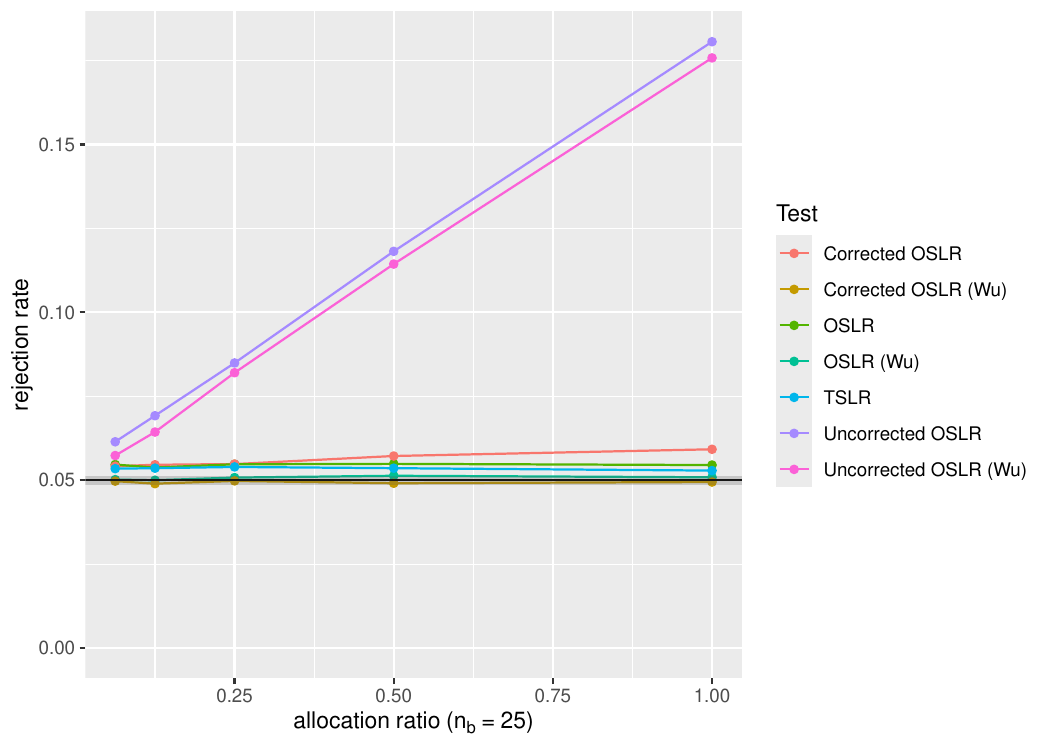"} &   \includegraphics[width=.44\textwidth]{"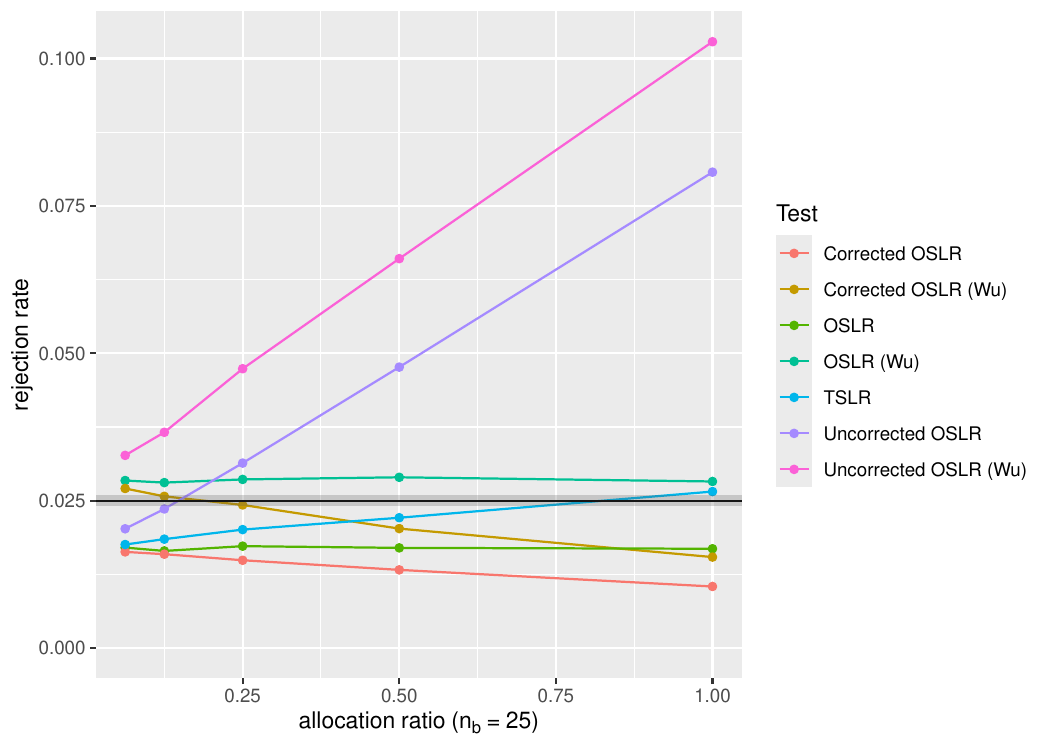"} \\
		\hline
		50 &\includegraphics[width=.44\textwidth]{"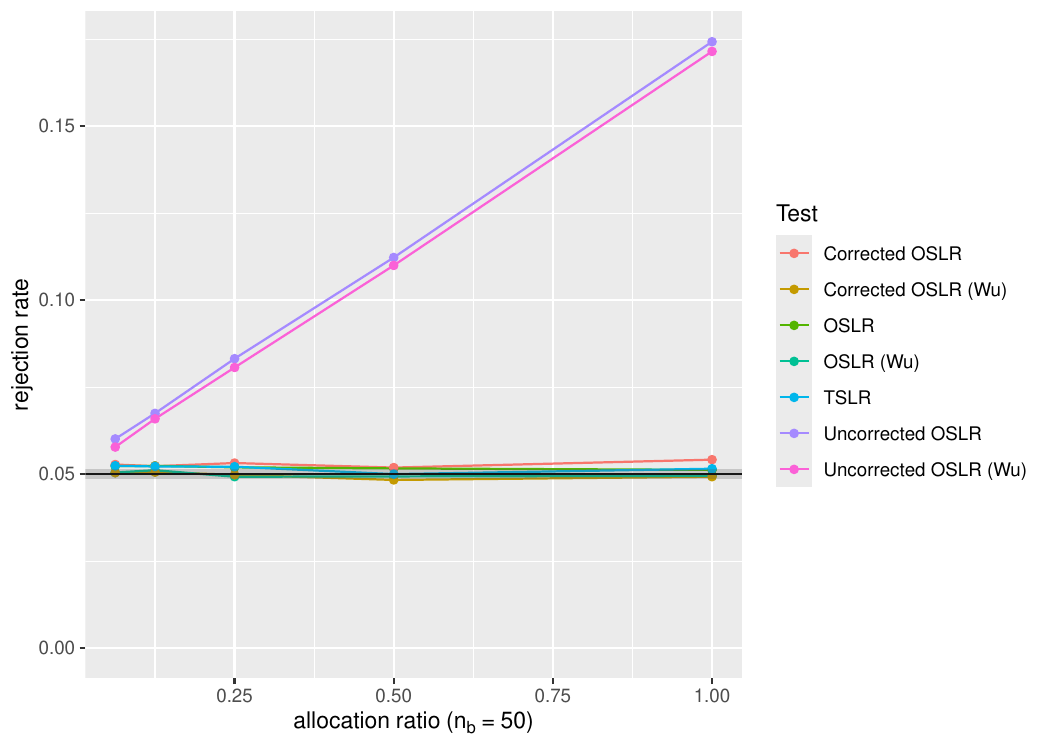"} &   \includegraphics[width=.44\textwidth]{"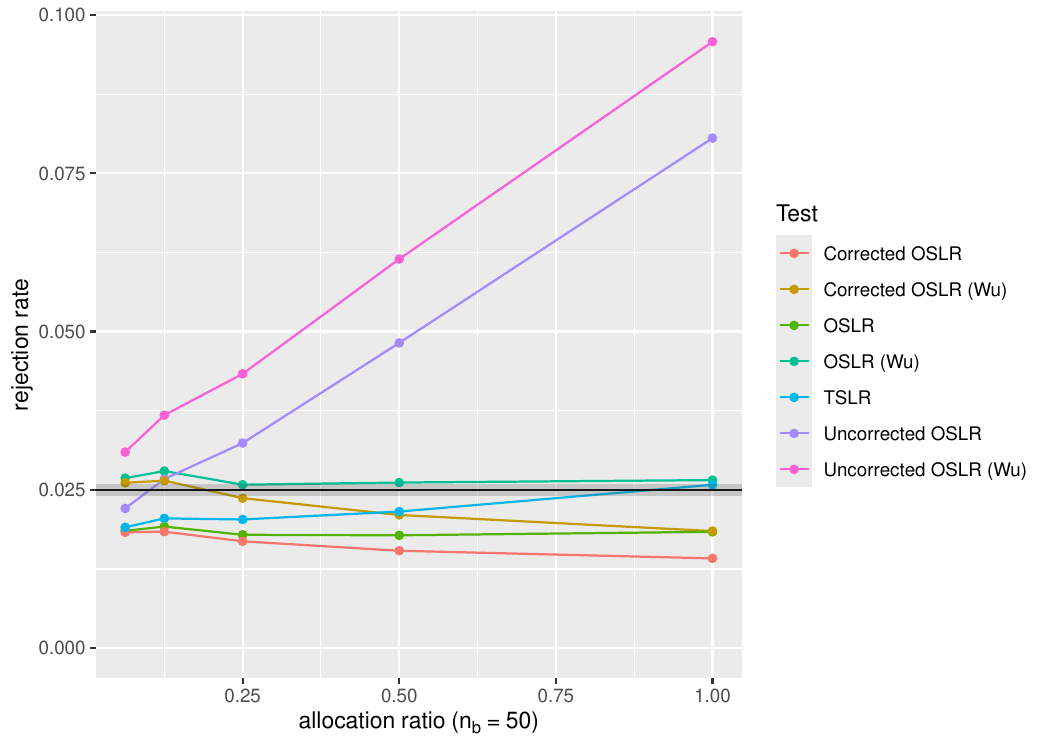"} \\
		\hline
		100 &\includegraphics[width=.44\textwidth]{"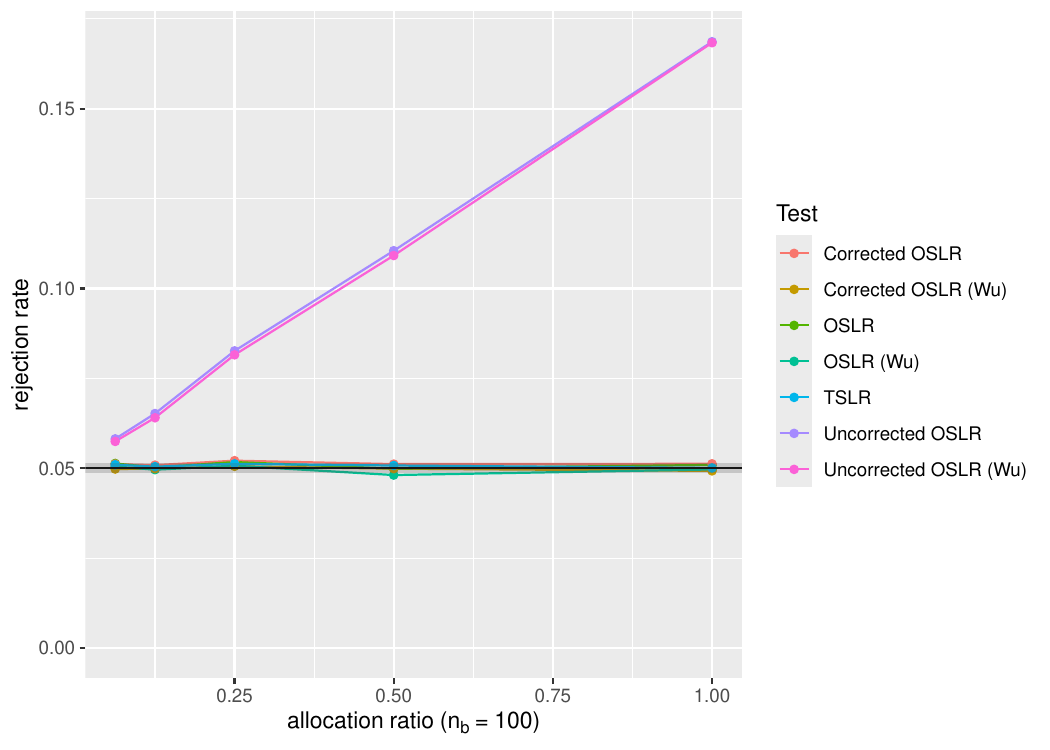"} &   \includegraphics[width=.44\textwidth]{"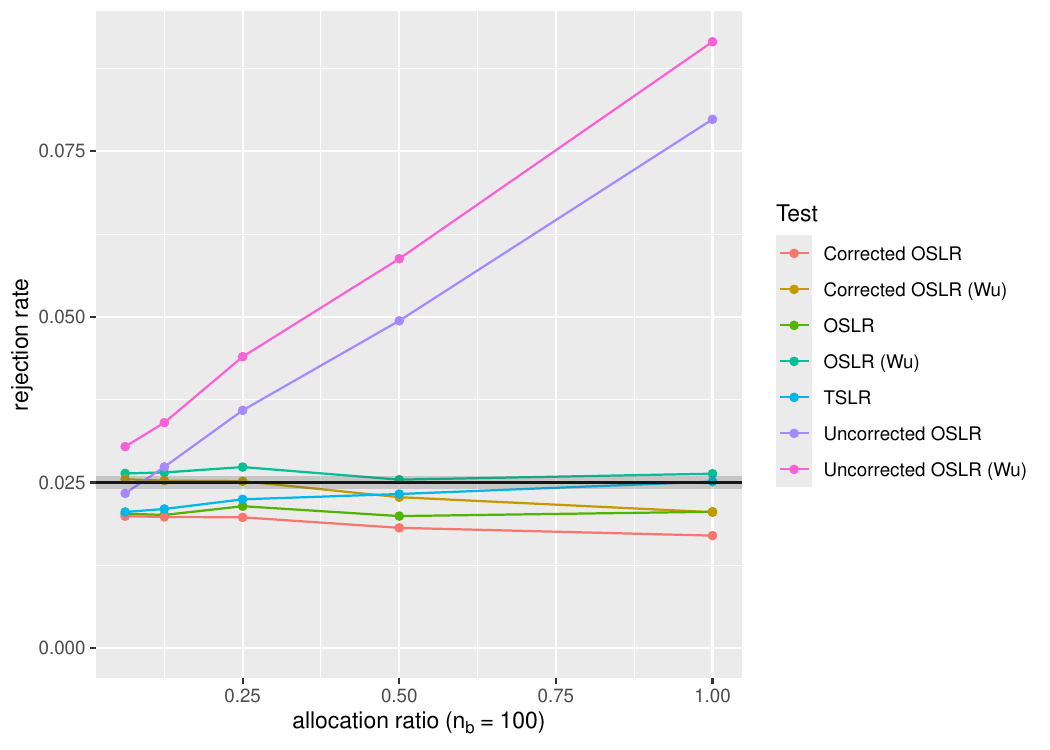"} \\
		\hline
		200 &\includegraphics[width=.44\textwidth]{"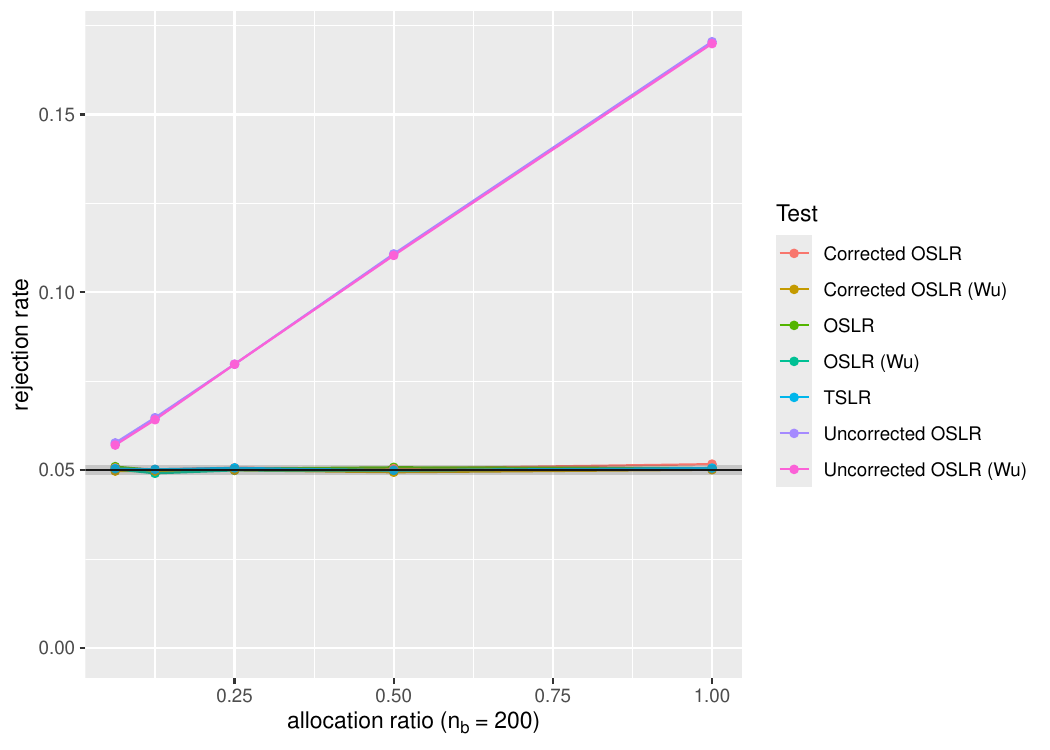"} &   \includegraphics[width=.44\textwidth]{"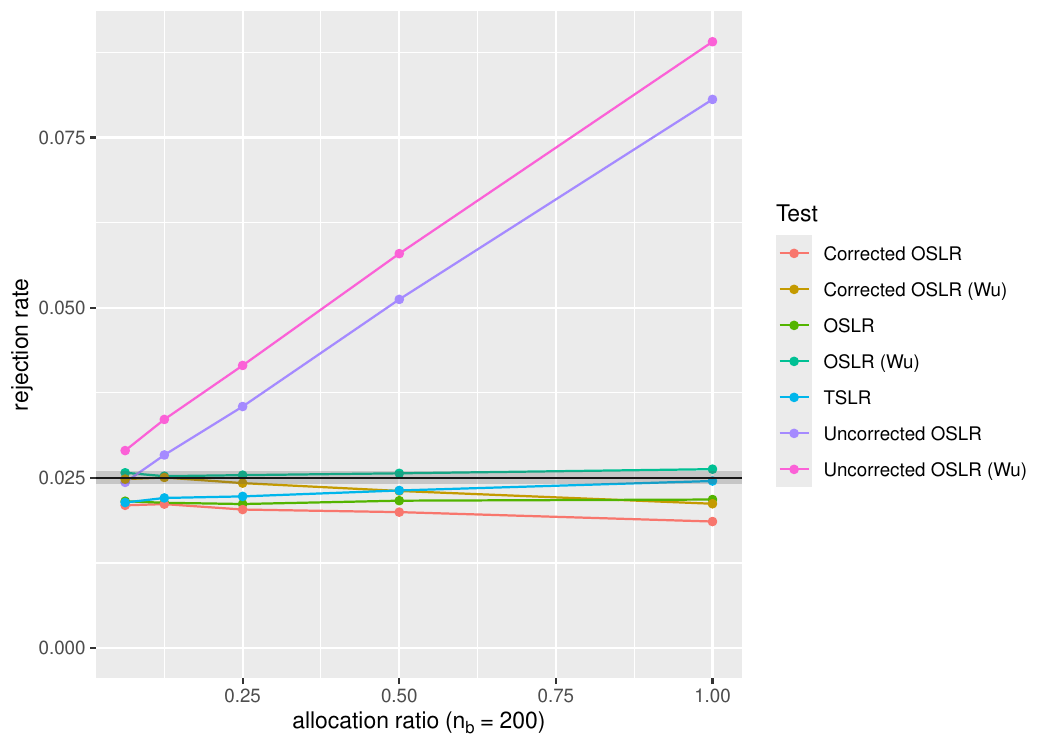"}
	\end{tabular}
	\caption{Two- and one-sided empirical rejection rates of the null hypothesis $H_0$ with Weibull-distributed data with shape parameter $\kappa = 0.5$ for four different sample sizes in the experimental cohort in dependence of the allocation ratio.}
\end{figure}

\begin{figure}[h!]
	\centering
	\begin{tabular}{c|c|c}
		$n_B$ & two-sided rates & one-sided rates \\
		\hline
		25 &\includegraphics[width=.44\textwidth]{"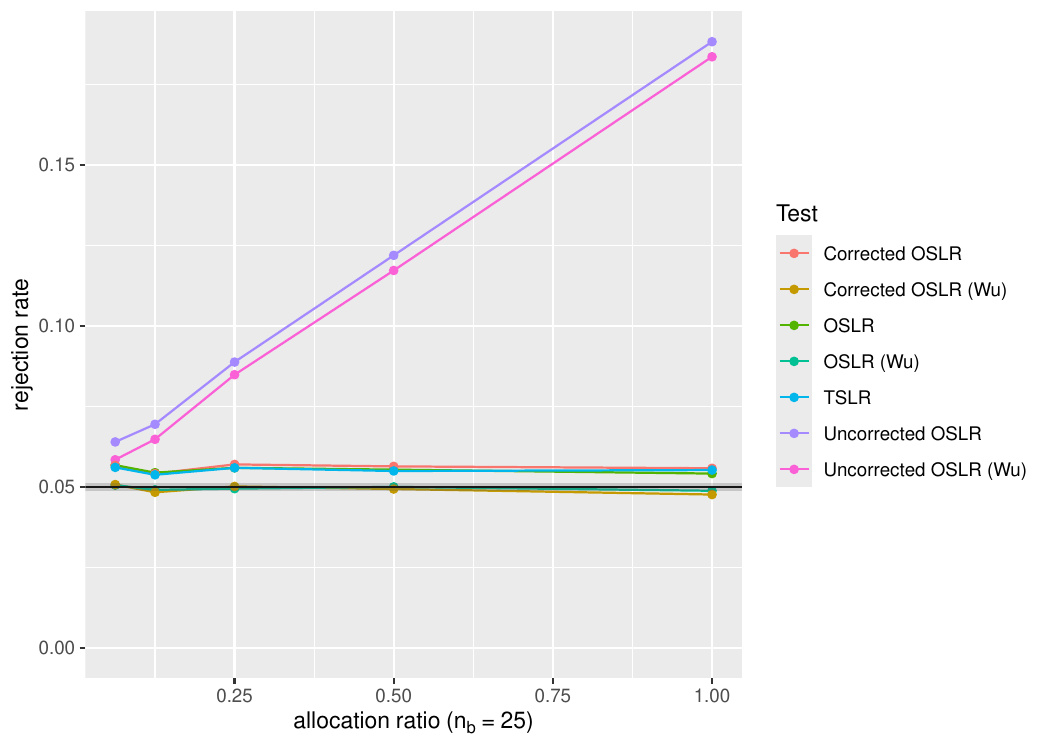"} &   \includegraphics[width=.44\textwidth]{"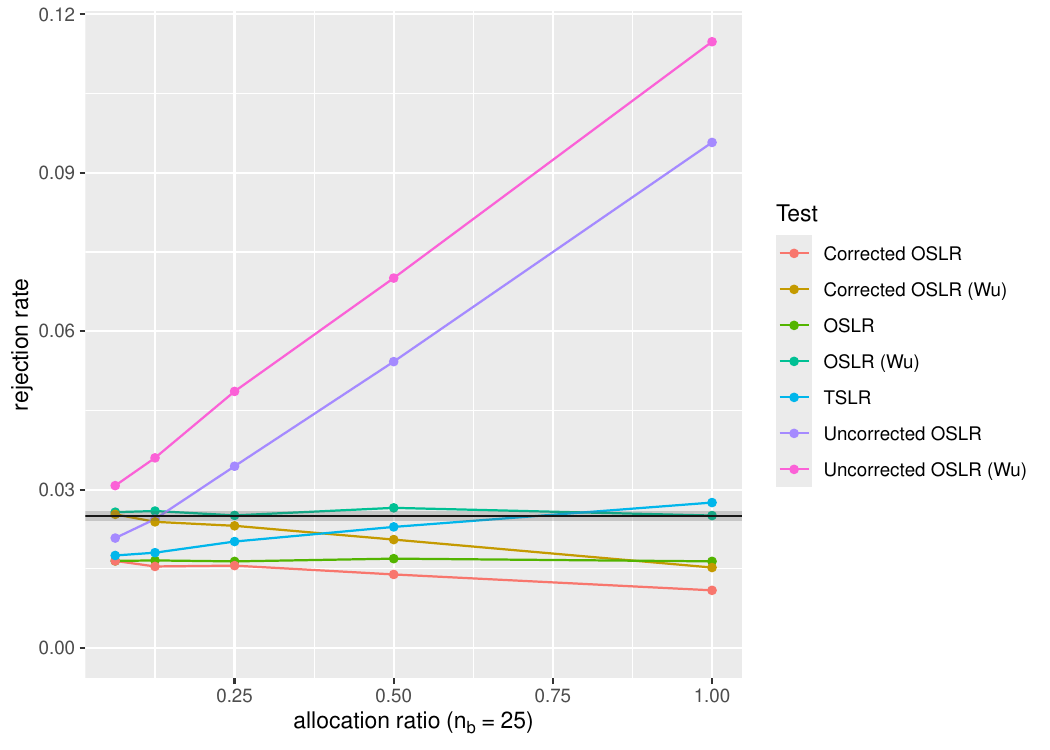"} \\
		\hline
		50 &\includegraphics[width=.44\textwidth]{"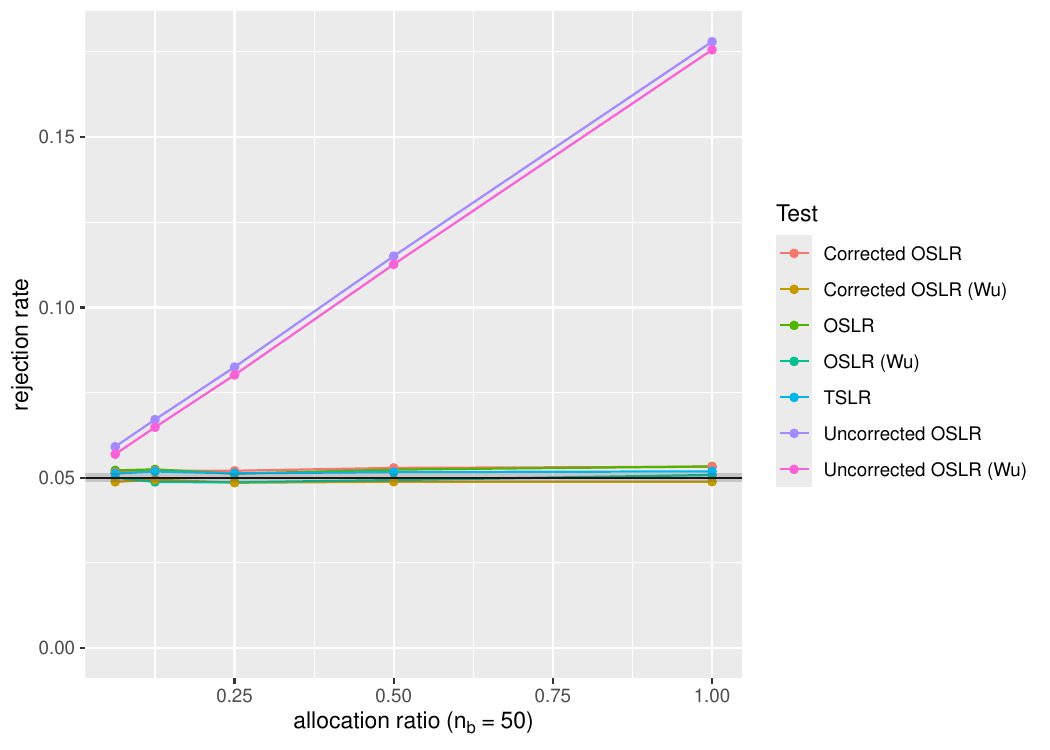"} &   \includegraphics[width=.44\textwidth]{"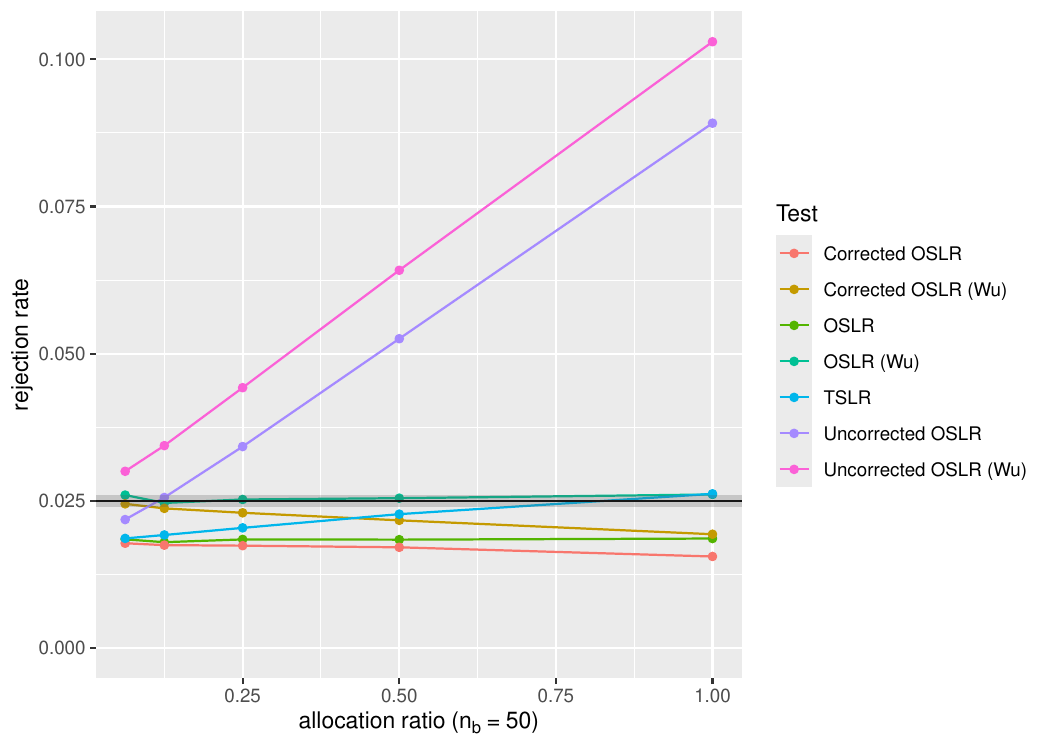"} \\
		\hline
		100 &\includegraphics[width=.44\textwidth]{"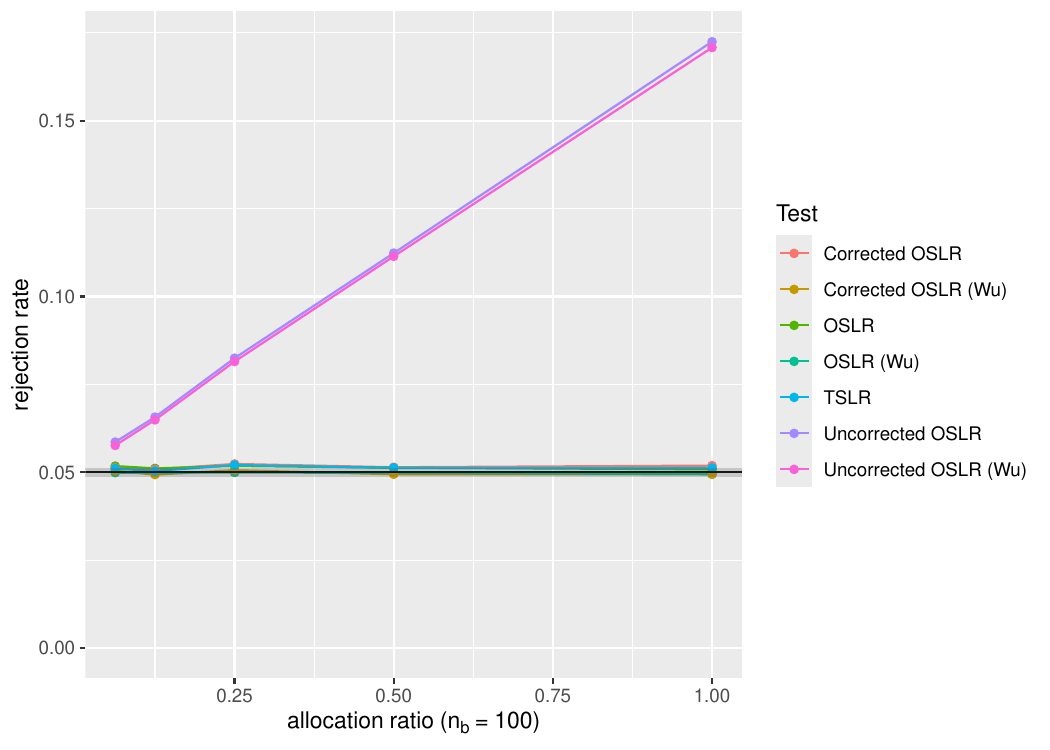"} &   \includegraphics[width=.44\textwidth]{"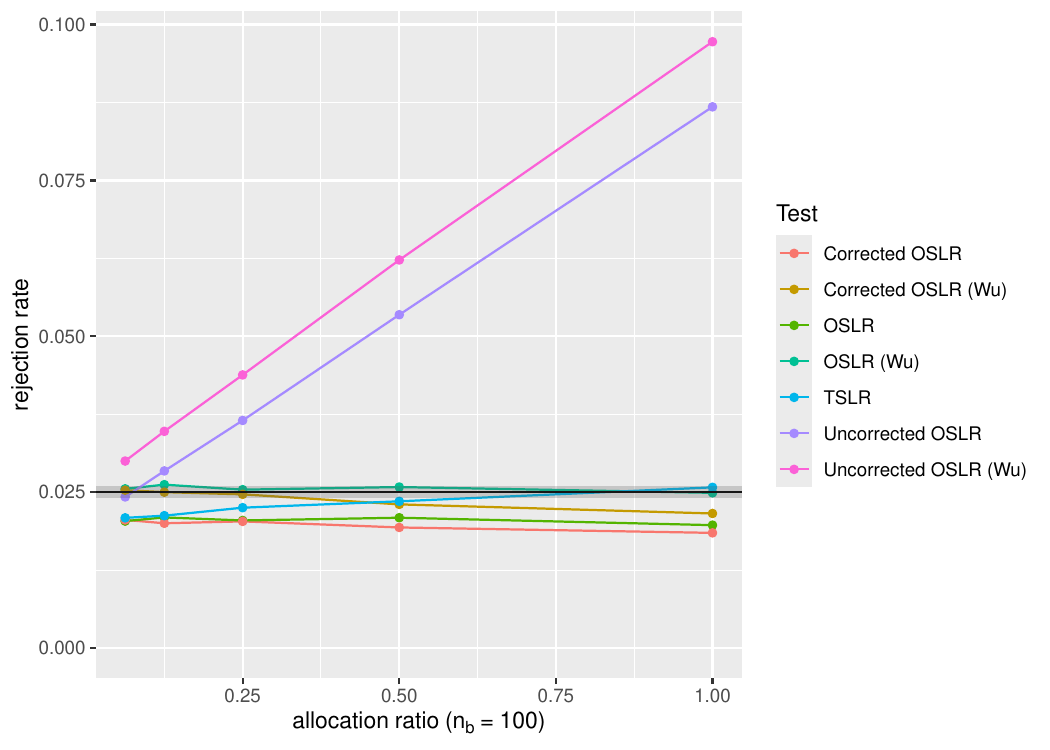"} \\
		\hline
		200 &\includegraphics[width=.44\textwidth]{"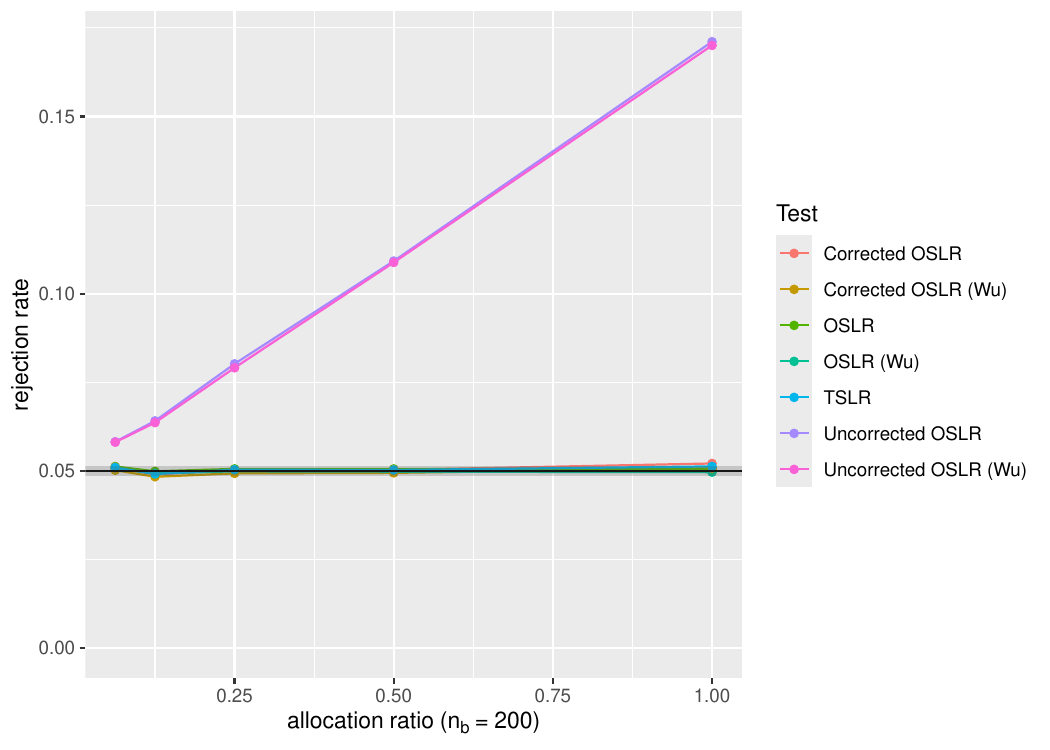"} &   \includegraphics[width=.44\textwidth]{"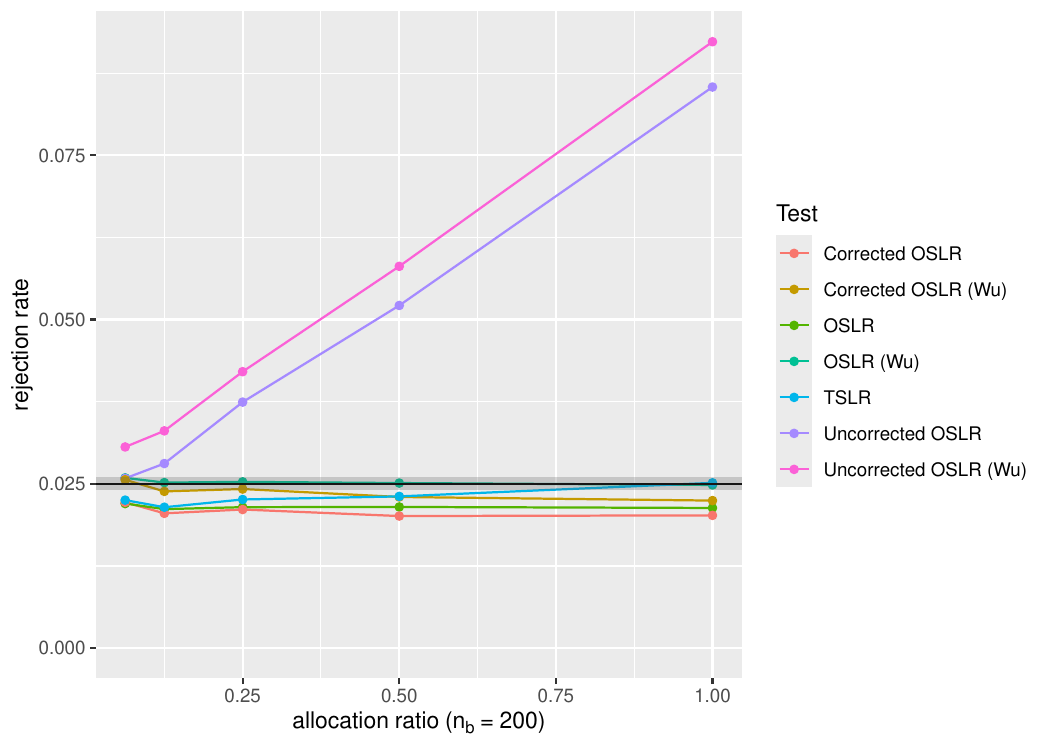"}
	\end{tabular}
	\caption{Two- and one-sided empirical rejection rates of the null hypothesis $H_0$ with Weibull-distributed data with shape parameter $\kappa = 1$ for four different sample sizes in the experimental cohort in dependence of the allocation ratio.}
\end{figure}

\begin{figure}[h!]
	\centering
	\begin{tabular}{c|c|c}
		$n_B$ & two-sided rates & one-sided rates \\
		\hline
		25 &\includegraphics[width=.44\textwidth]{"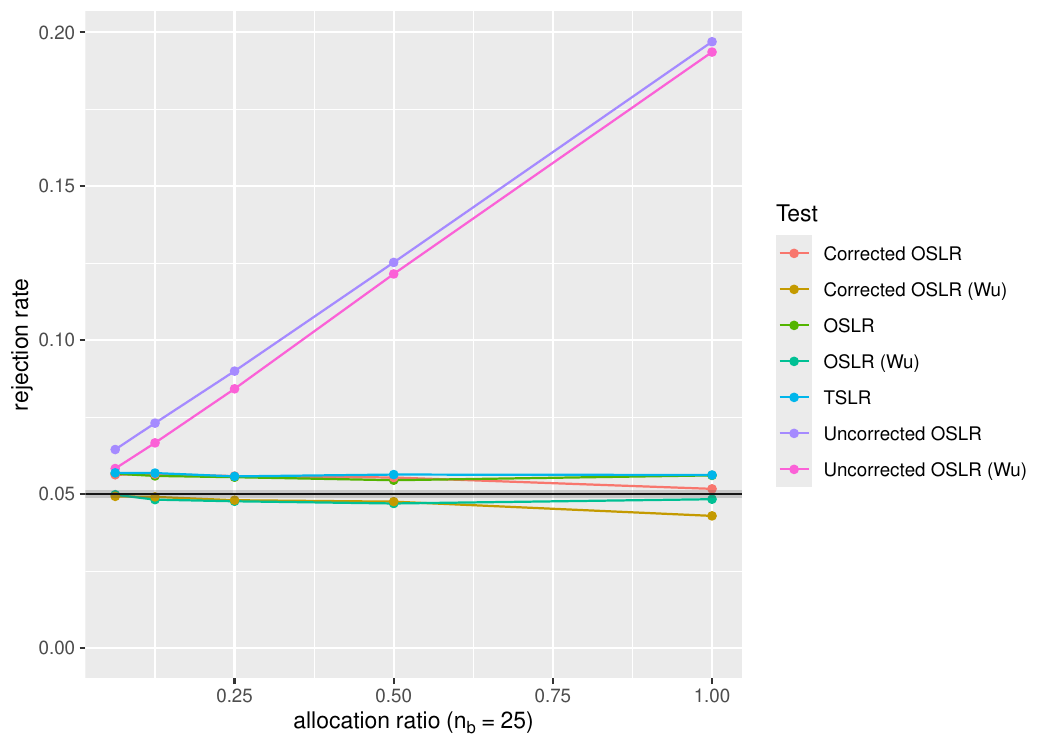"} &   \includegraphics[width=.44\textwidth]{"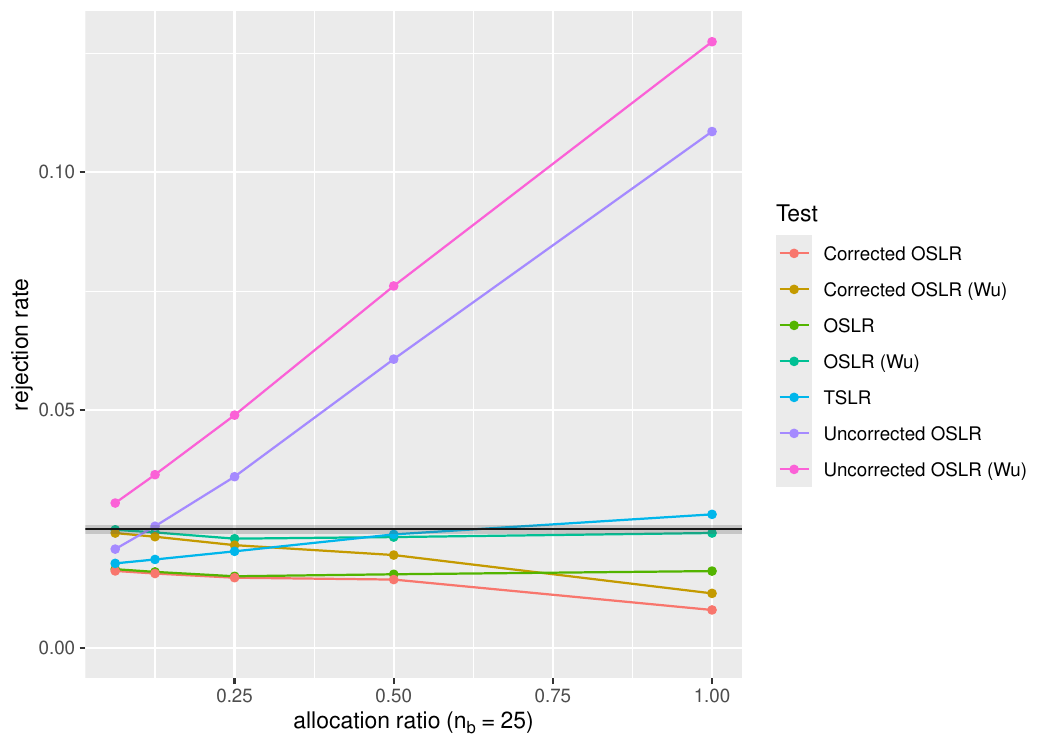"} \\
		\hline
		50 &\includegraphics[width=.44\textwidth]{"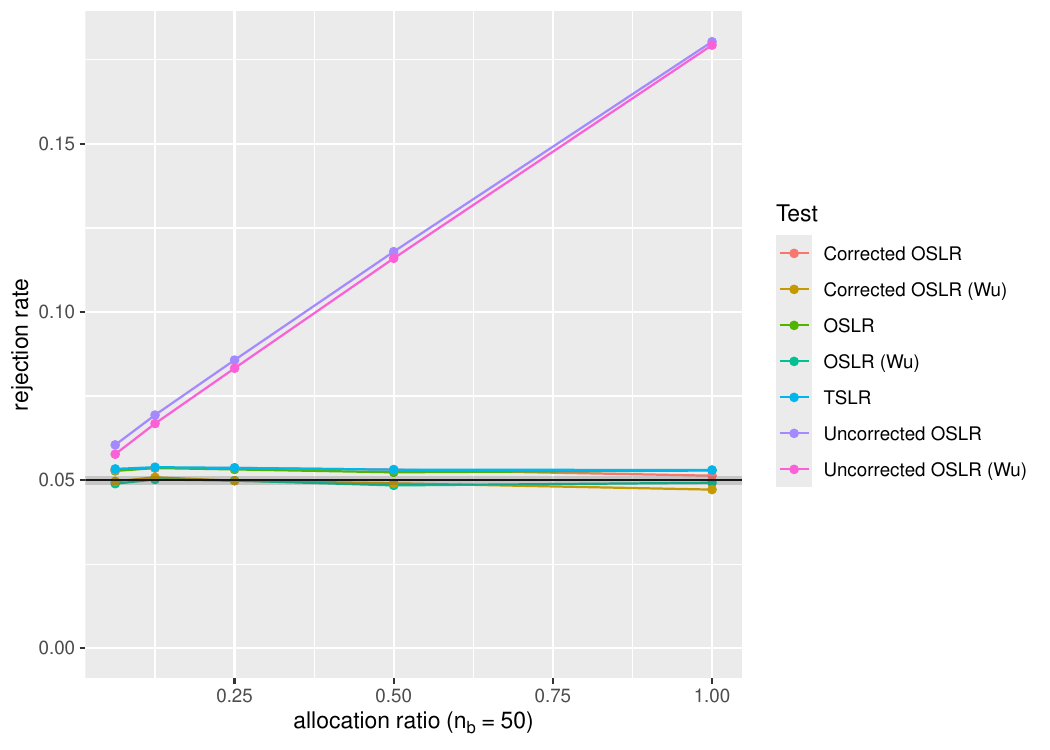"} &   \includegraphics[width=.44\textwidth]{"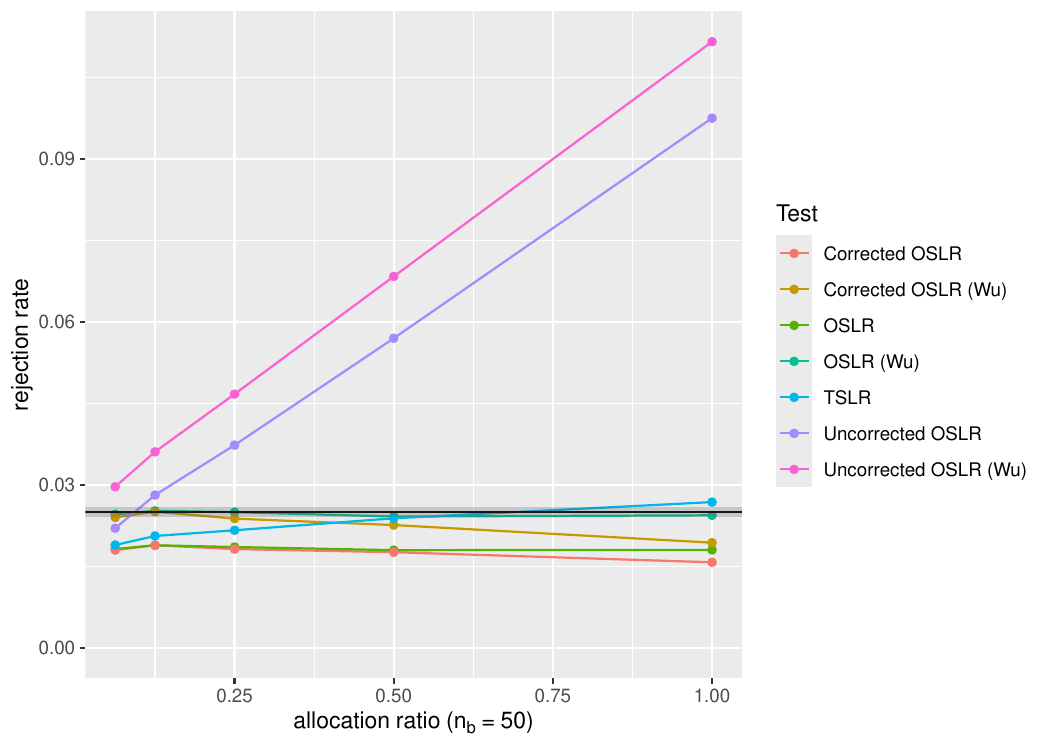"} \\
		\hline
		100 &\includegraphics[width=.44\textwidth]{"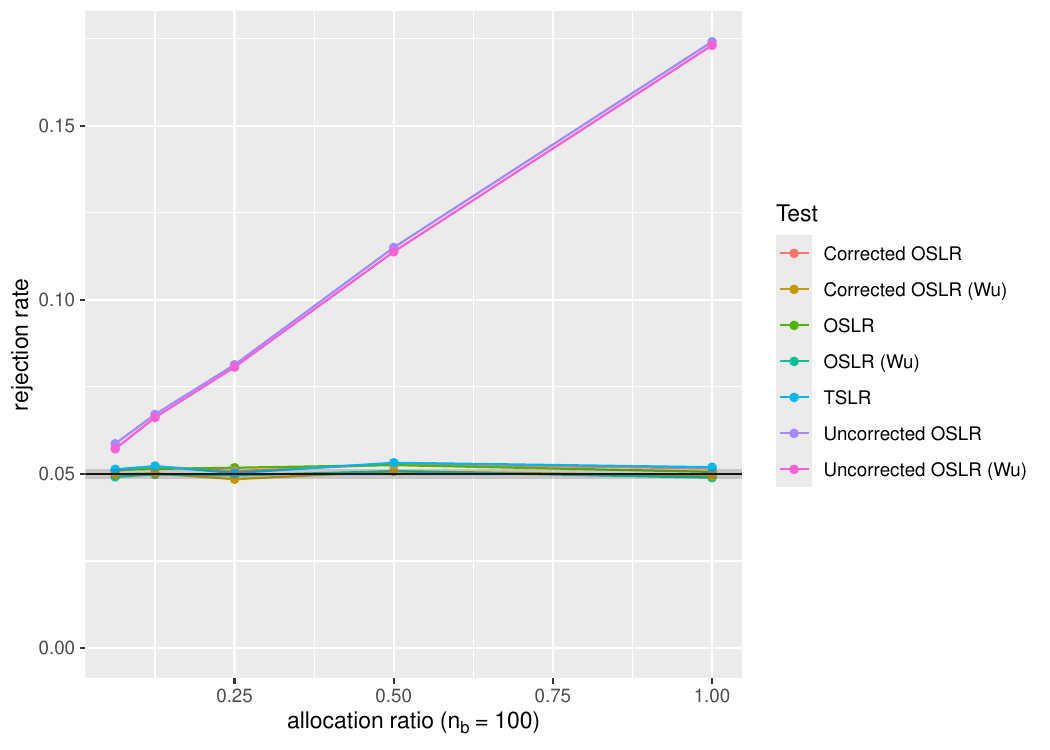"} &   \includegraphics[width=.44\textwidth]{"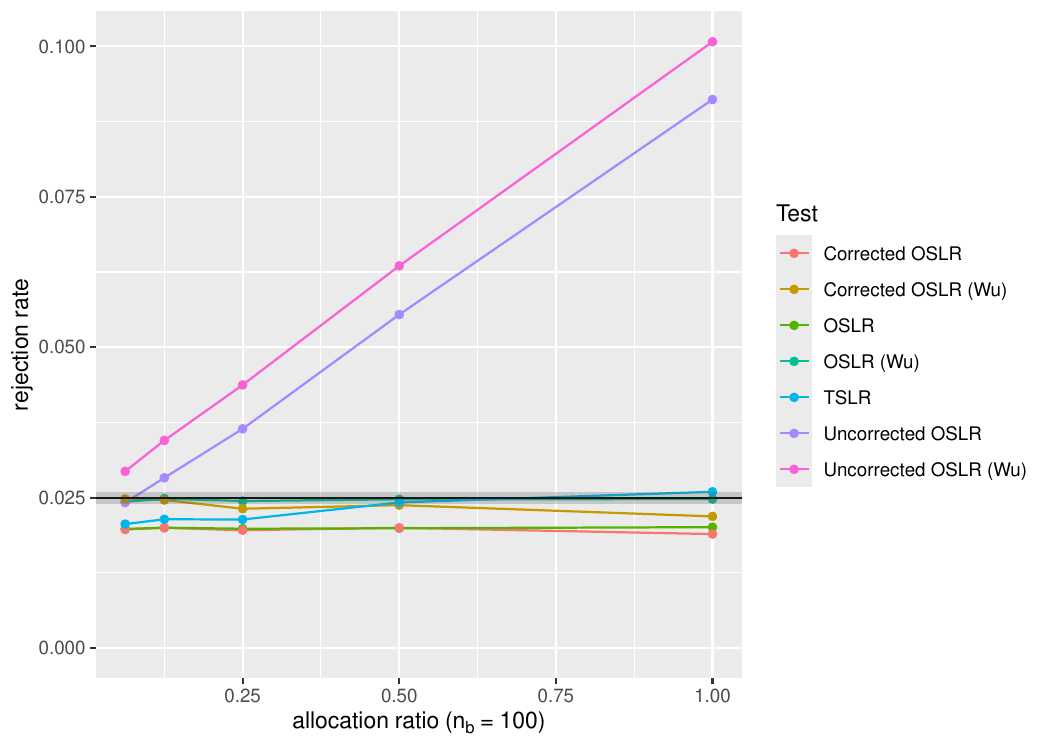"} \\
		\hline
		200 &\includegraphics[width=.44\textwidth]{"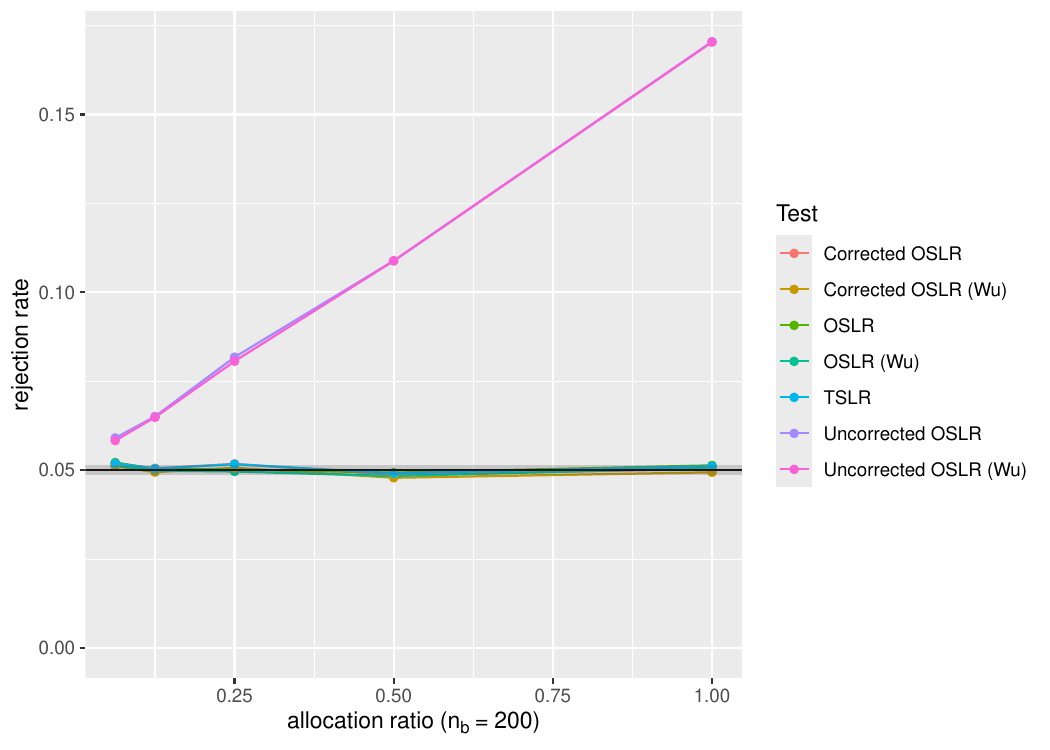"} &   \includegraphics[width=.44\textwidth]{"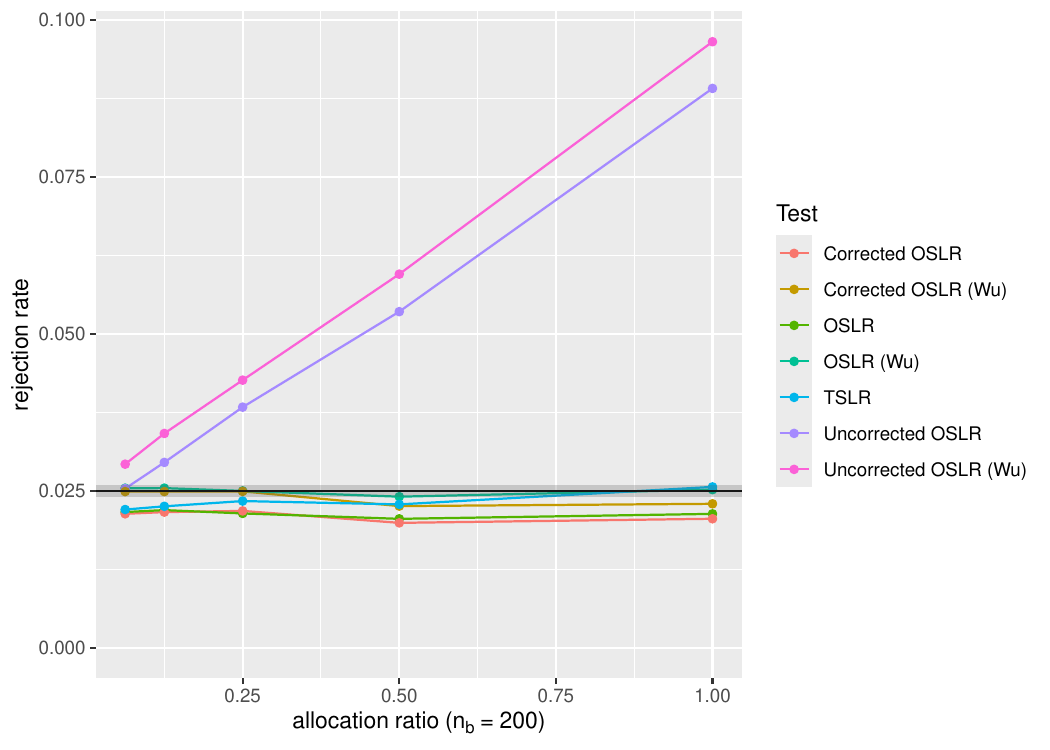"}
	\end{tabular}
	\caption{Two- and one-sided empirical rejection rates of the null hypothesis $H_0$ with Weibull-distributed data with shape parameter $\kappa = 2$ for four different sample sizes in the experimental cohort in dependence of the allocation ratio.}
\end{figure}

\clearpage

\subsection{Type I error rate (Exponential and Weibull MLEs)}

\begin{figure}[h!]
	\centering
	\begin{tabular}{c|c|c}
		$n_B$ & two-sided rates & one-sided rates \\
		\hline
		25 &\includegraphics[width=.44\textwidth]{"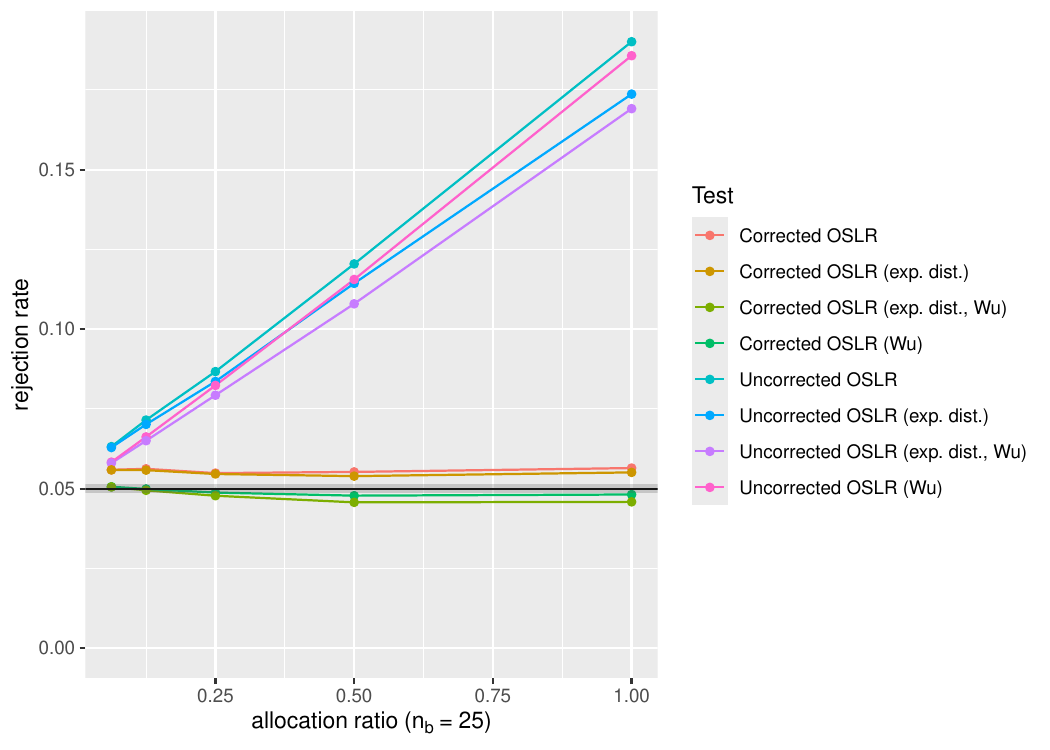"} &   \includegraphics[width=.44\textwidth]{"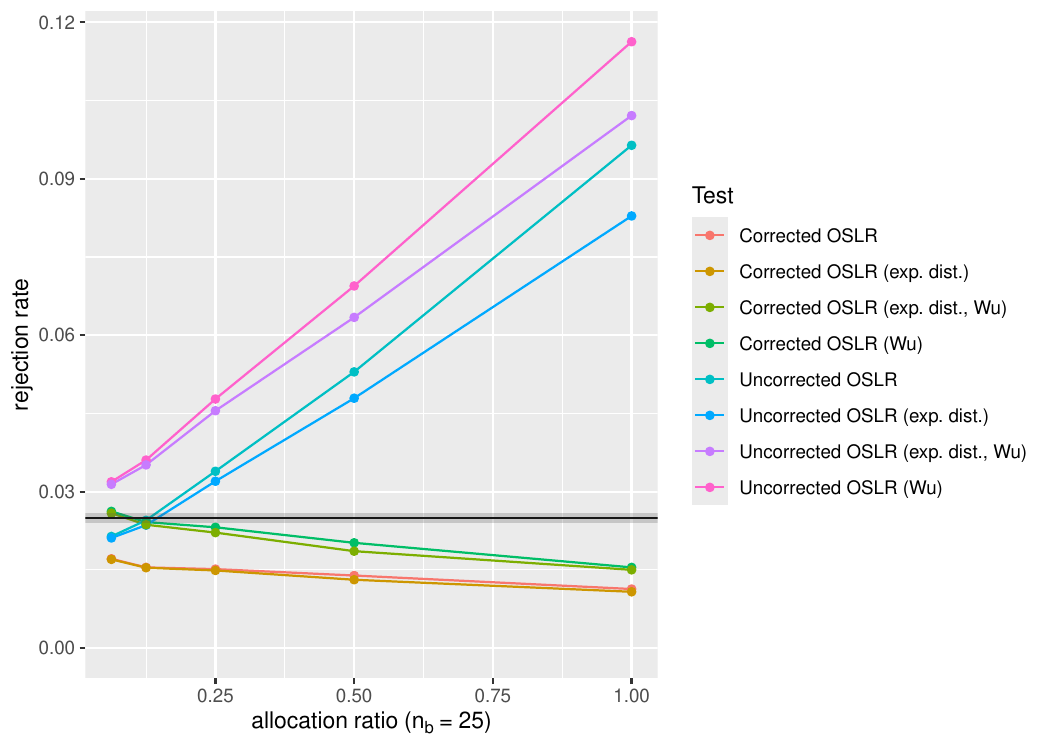"} \\
		\hline
		50 &\includegraphics[width=.44\textwidth]{"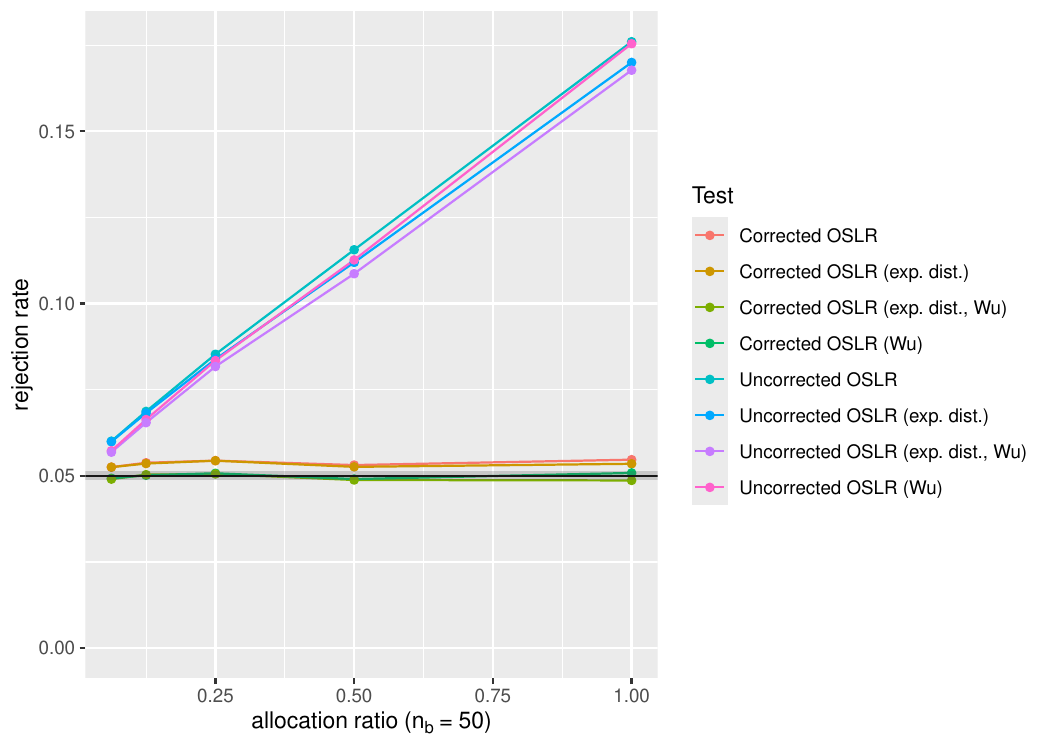"} &   \includegraphics[width=.44\textwidth]{"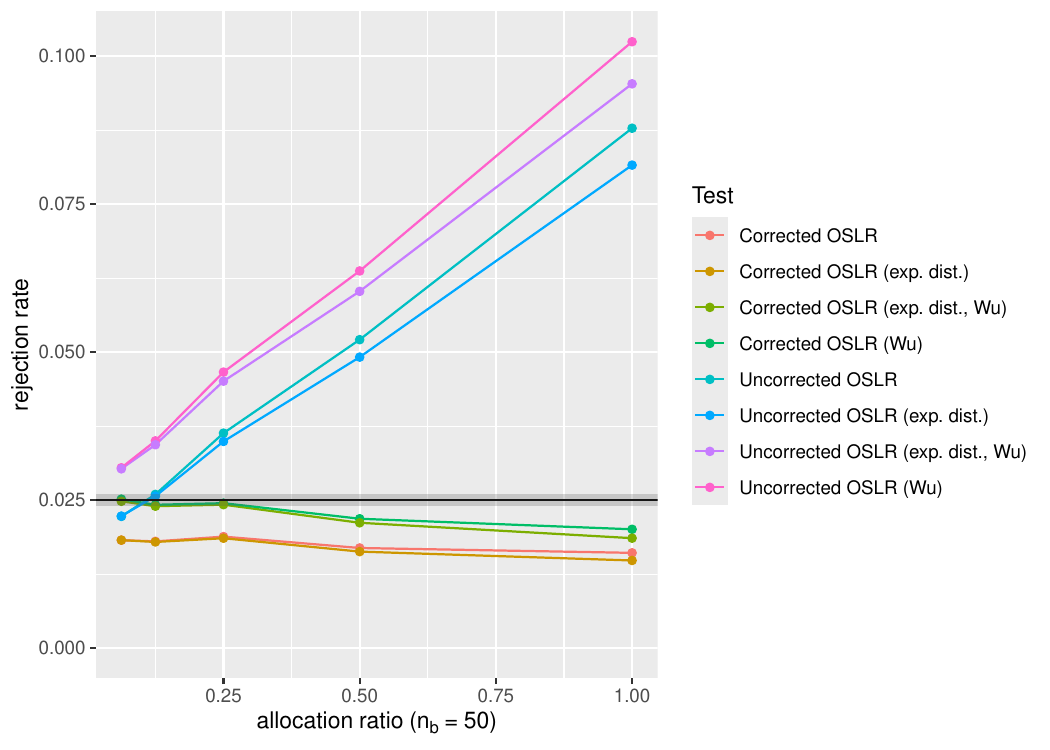"} \\
		\hline
		100 &\includegraphics[width=.44\textwidth]{"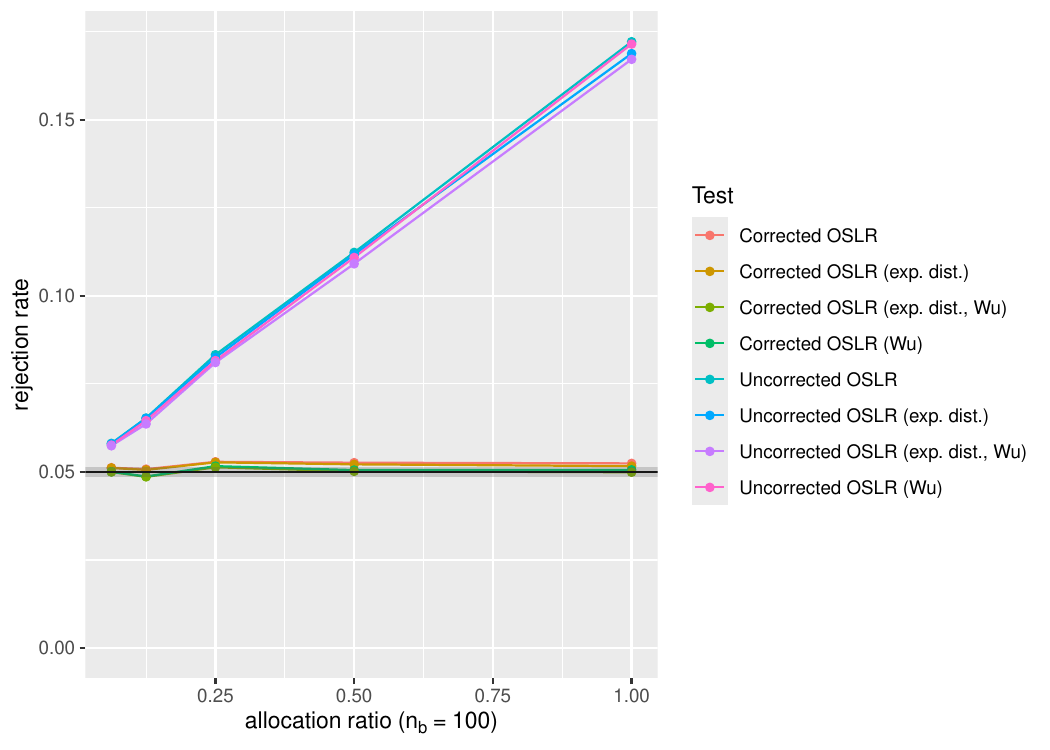"} &   \includegraphics[width=.44\textwidth]{"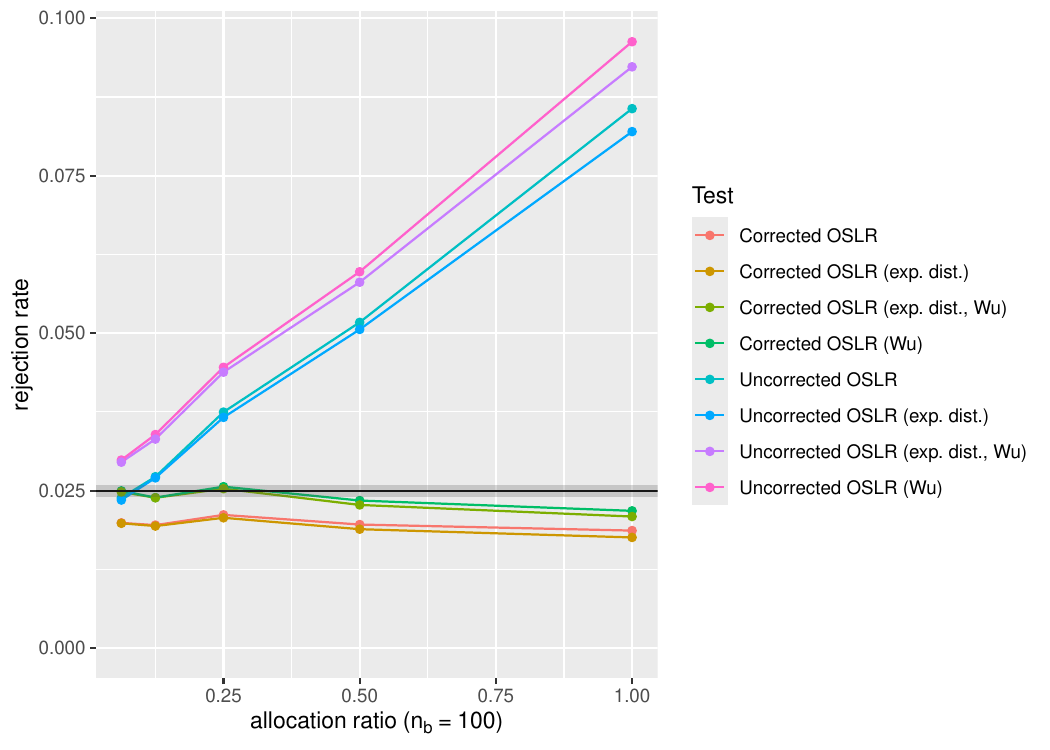"} \\
		\hline
		200 &\includegraphics[width=.44\textwidth]{"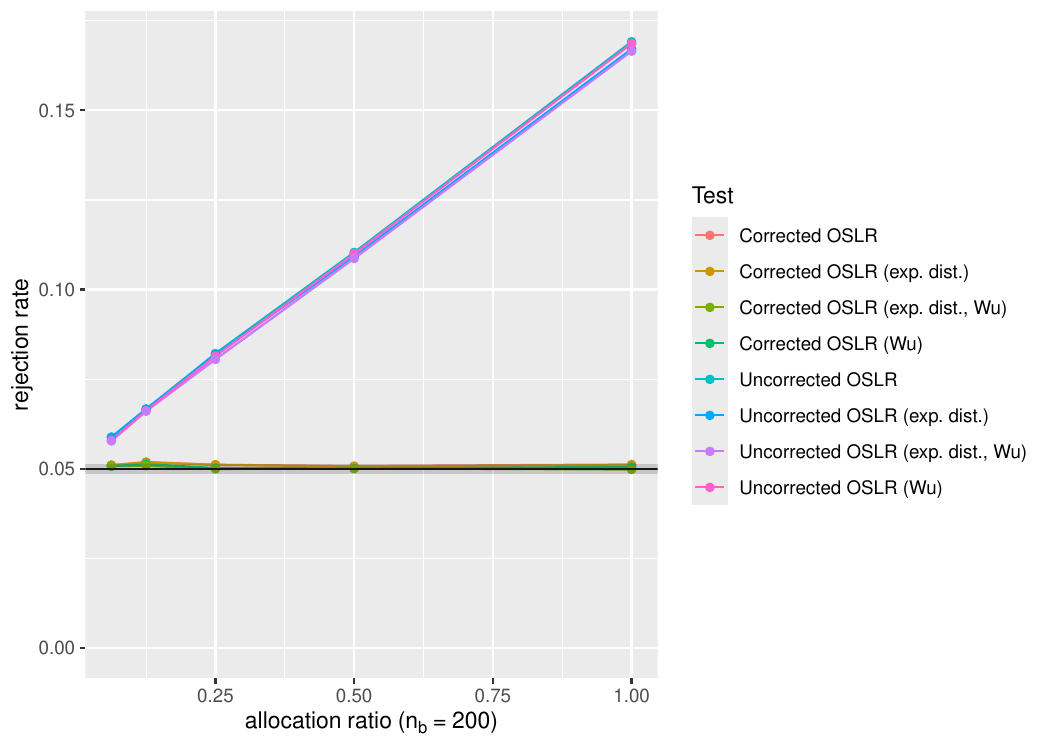"} &   \includegraphics[width=.44\textwidth]{"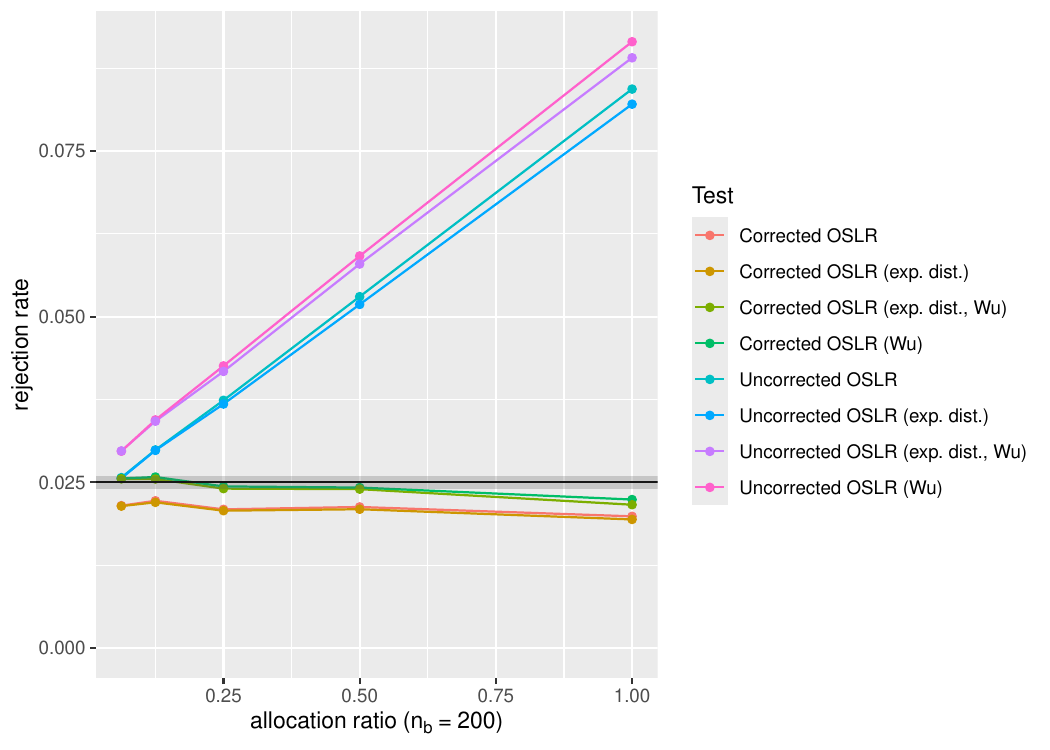"}
	\end{tabular}
	\caption{Two- and one-sided empirical rejection rates of the null hypothesis $H_0$ with exponentially distributed data for procedures based on exponential and Weibull MLEs for four different sample sizes in the experimental cohort in dependence of the allocation ratio.}
\end{figure}

\clearpage

\subsection{Power}

\begin{figure}[h!]
	\centering
	\begin{tabular}{c|c}
		$n_B$ & power \\
		\hline
		25 &\includegraphics[width=.44\textwidth]{"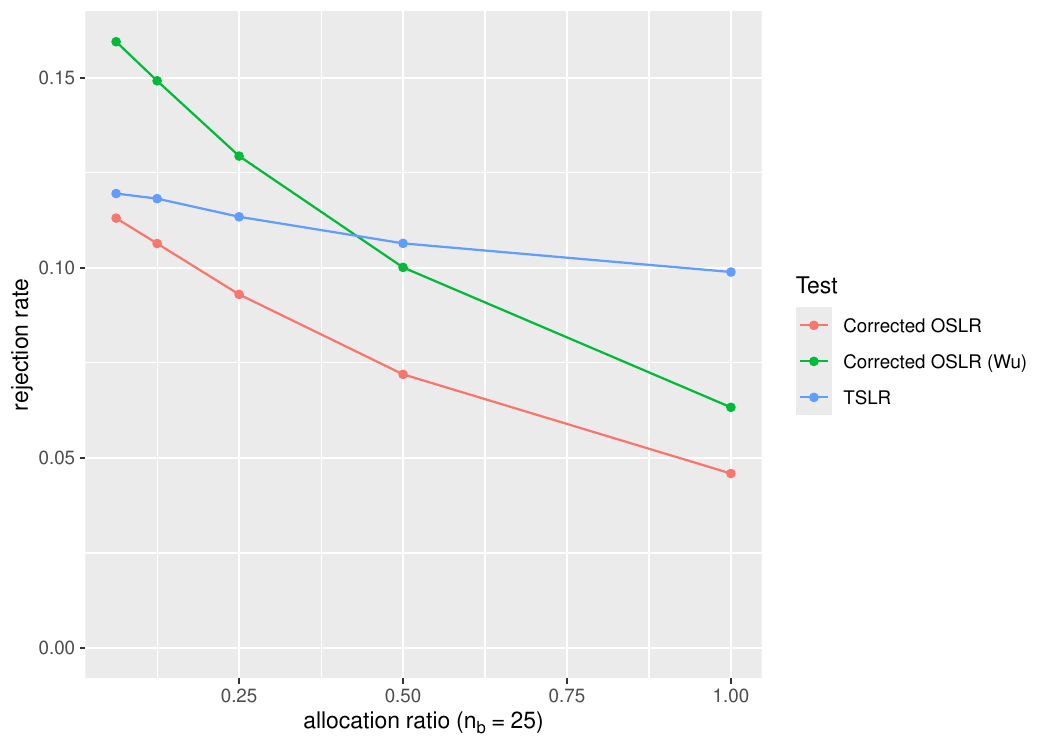"} \\
		\hline
		50 &\includegraphics[width=.44\textwidth]{"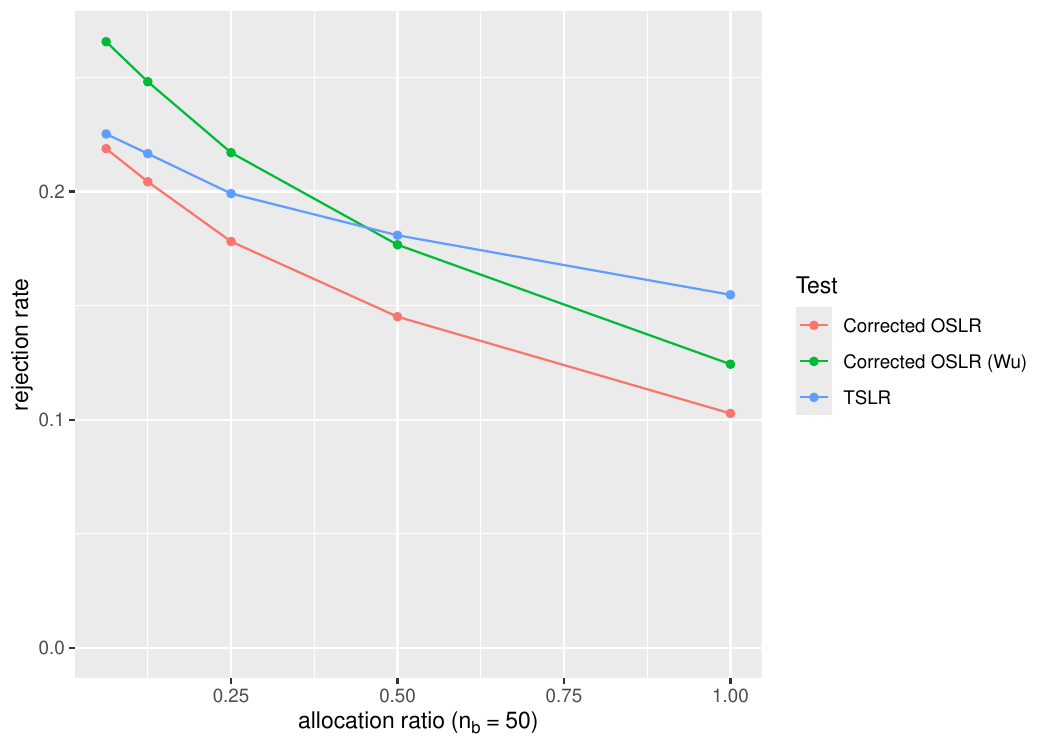"} \\
		\hline
		100 &\includegraphics[width=.44\textwidth]{"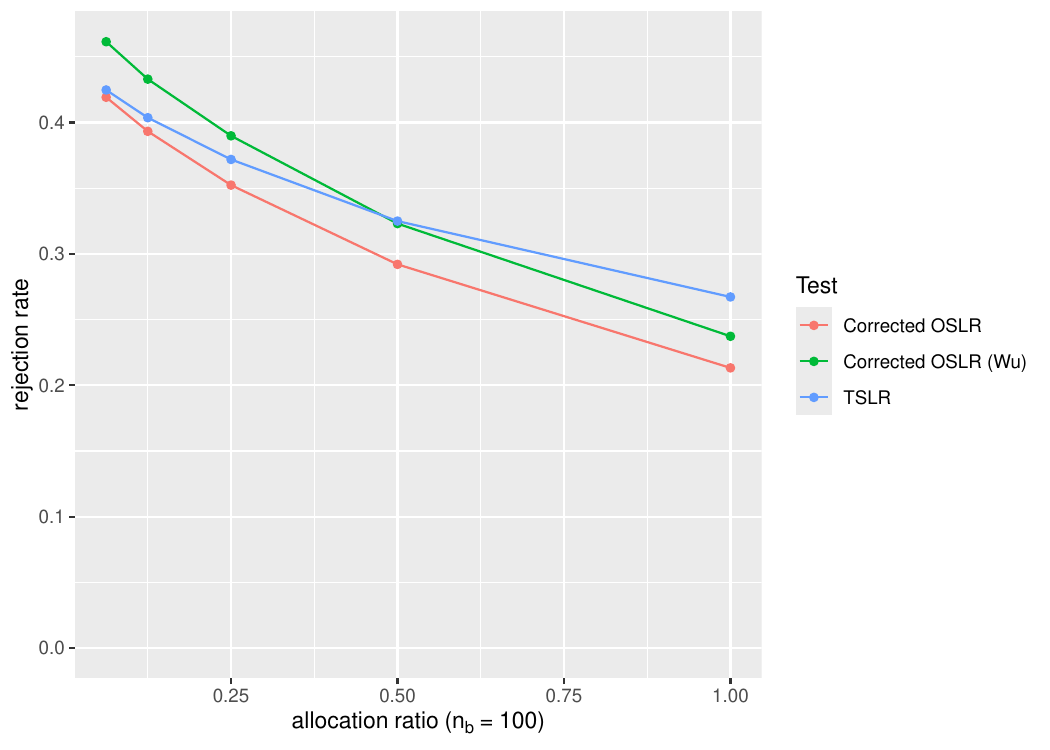"} \\
		\hline
		200 &\includegraphics[width=.44\textwidth]{"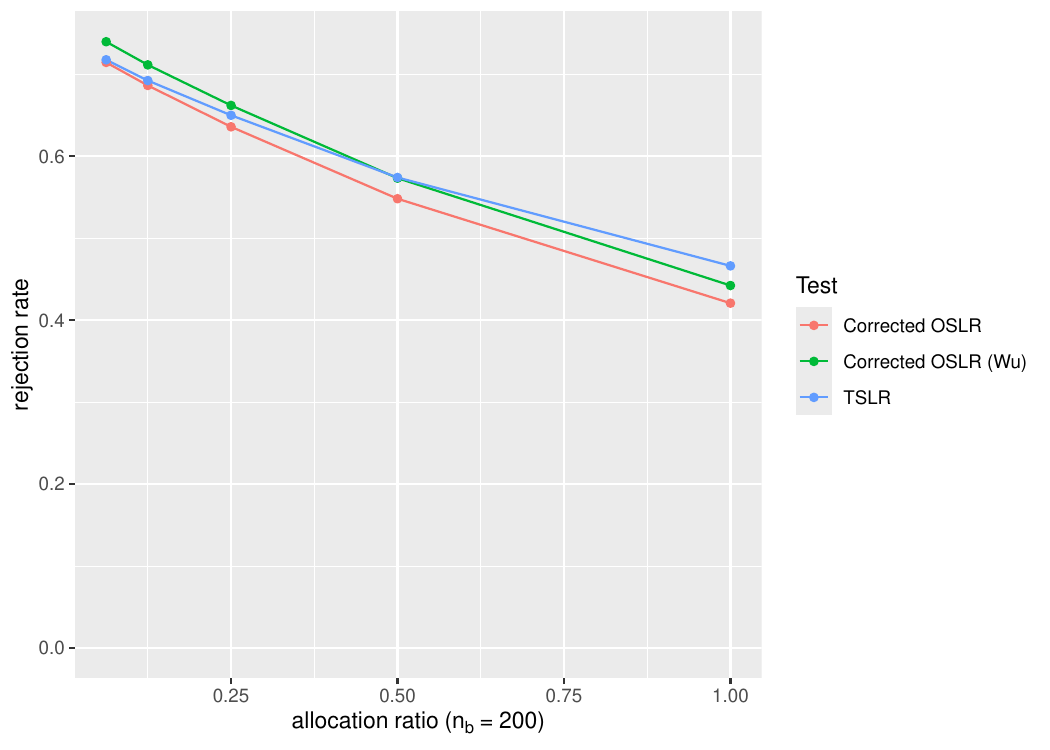"} 
	\end{tabular}
	\caption{Power to reject the null hypothesis $H_0$ for Weibull-distributed data with shape parameter $\kappa = 0.5$ and hazard ratio $\omega= 0.8$ for four different sample sizes in the experimental cohort in dependence of the allocation ratio.}
\end{figure}

\begin{figure}[h!]
	\centering
	\begin{tabular}{c|c|c}
		$n_B$ & power (excl. procedures with exponential MLEs) & power (incl. procedures with exponential MLEs) \\
		\hline
		25 &\includegraphics[width=.44\textwidth]{"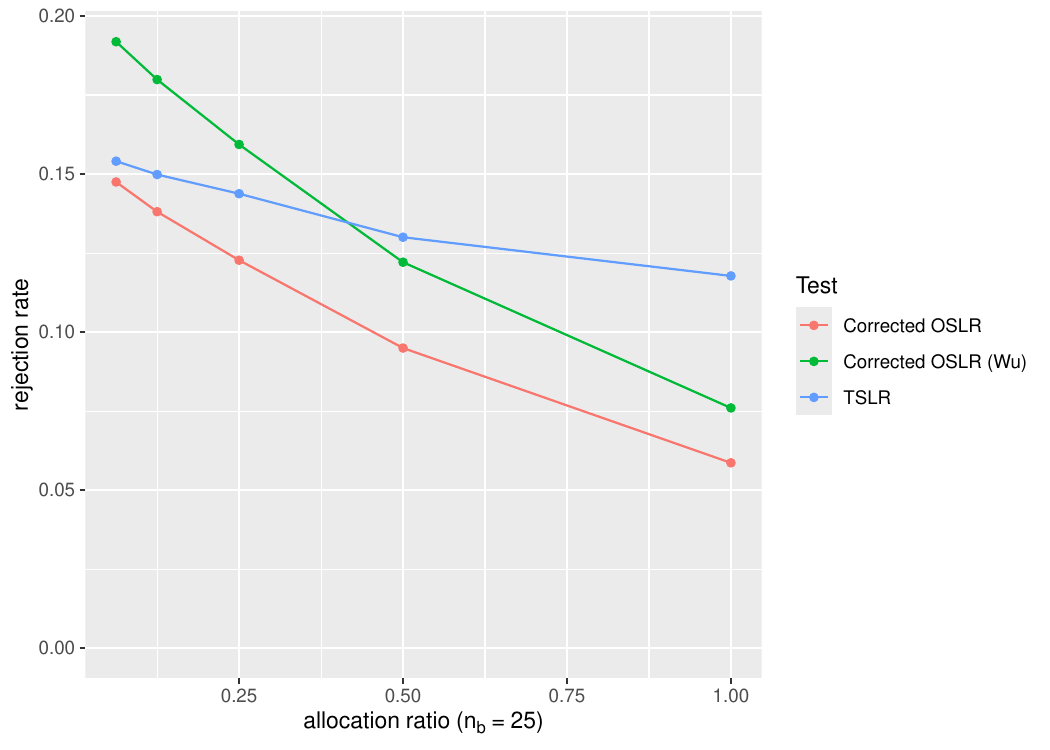"} &\includegraphics[width=.44\textwidth]{"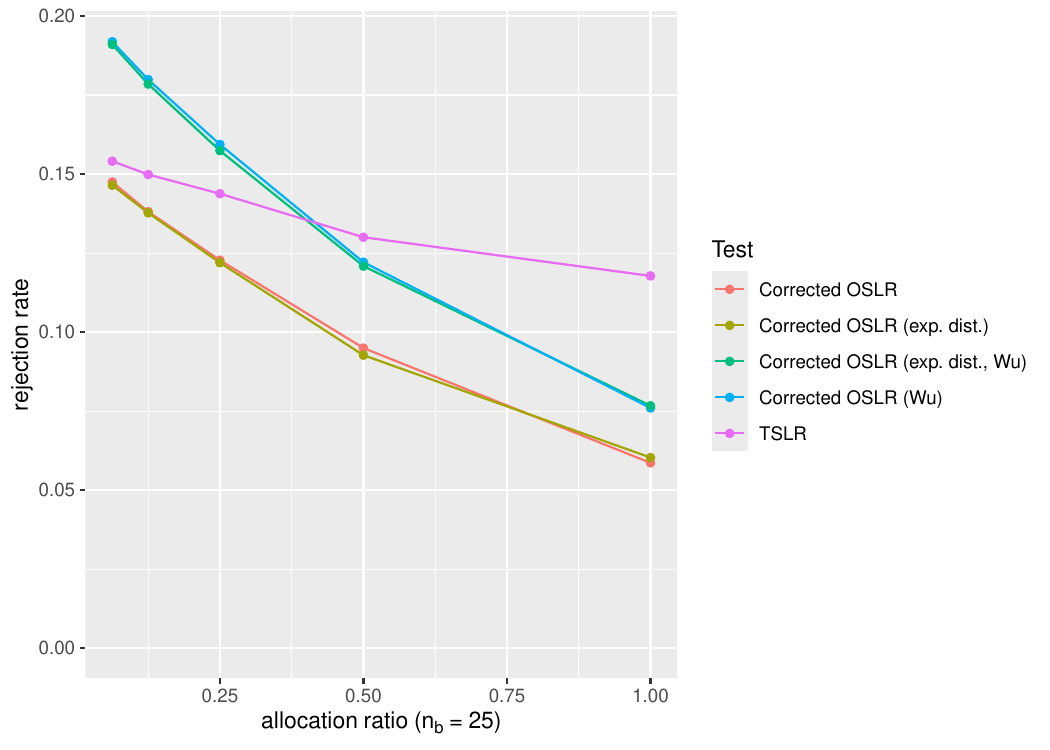"} \\
		\hline
		50 &\includegraphics[width=.44\textwidth]{"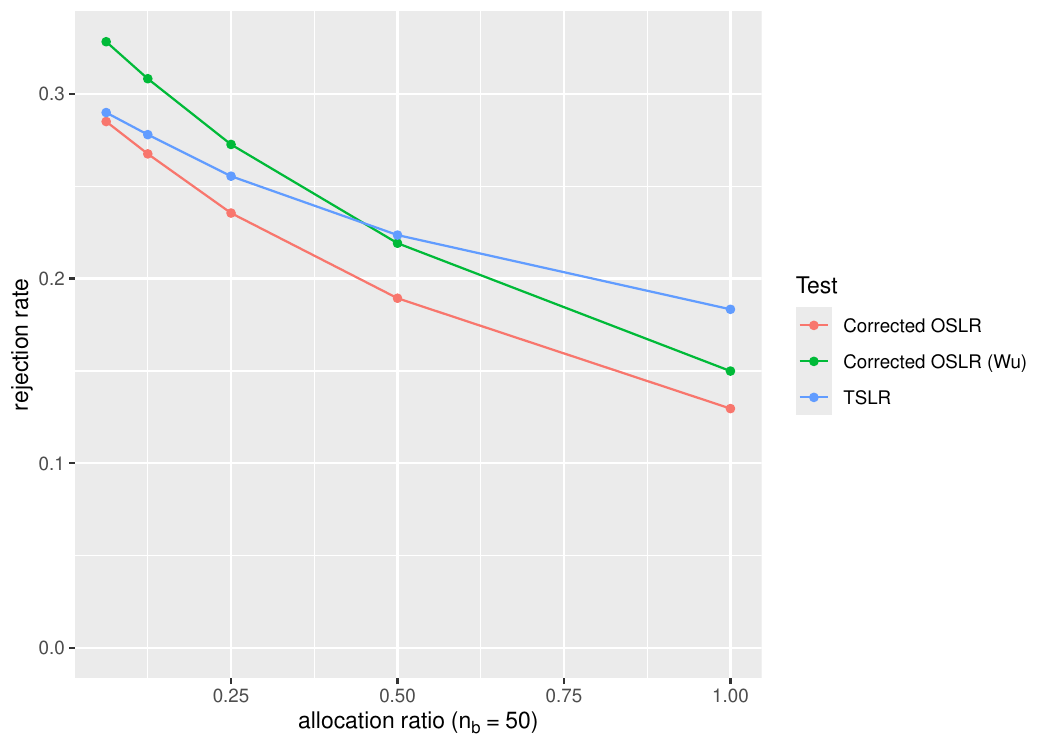"} &\includegraphics[width=.44\textwidth]{"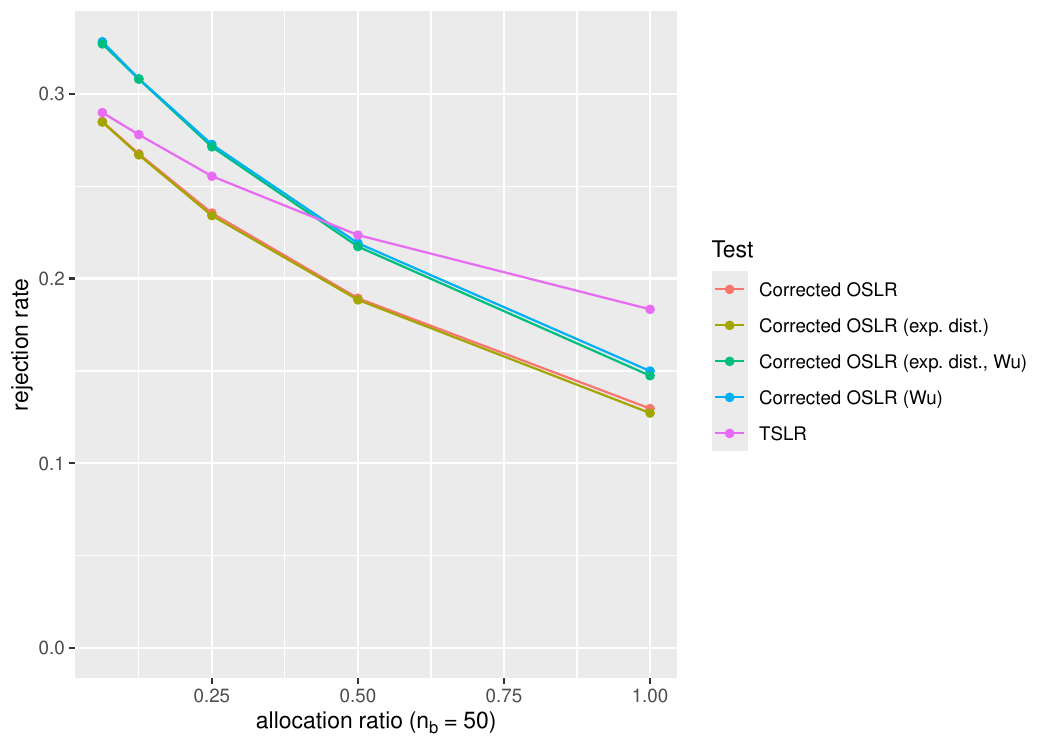"} \\
		\hline
		100 &\includegraphics[width=.44\textwidth]{"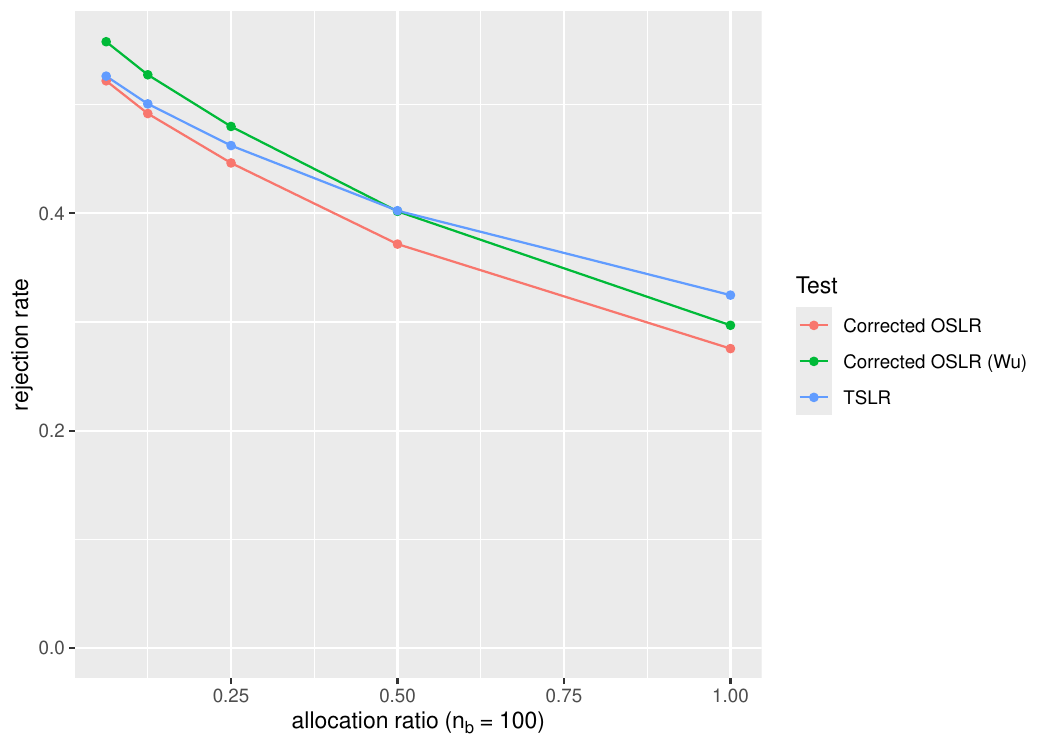"} &\includegraphics[width=.44\textwidth]{"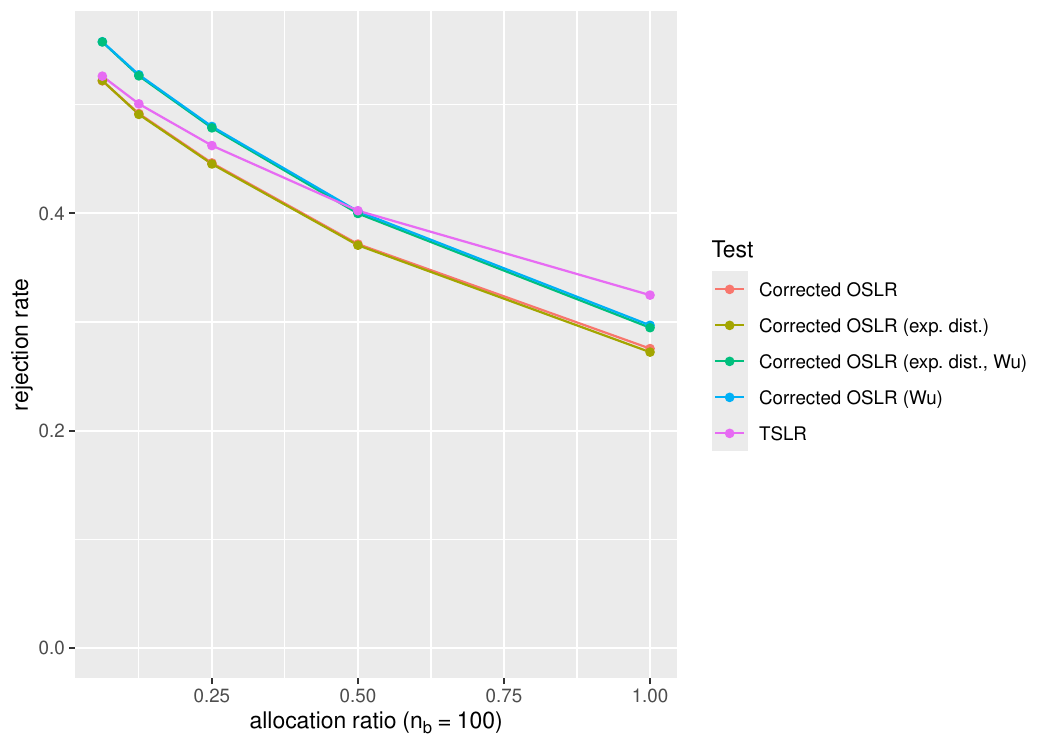"} \\
		\hline
		200 &\includegraphics[width=.44\textwidth]{"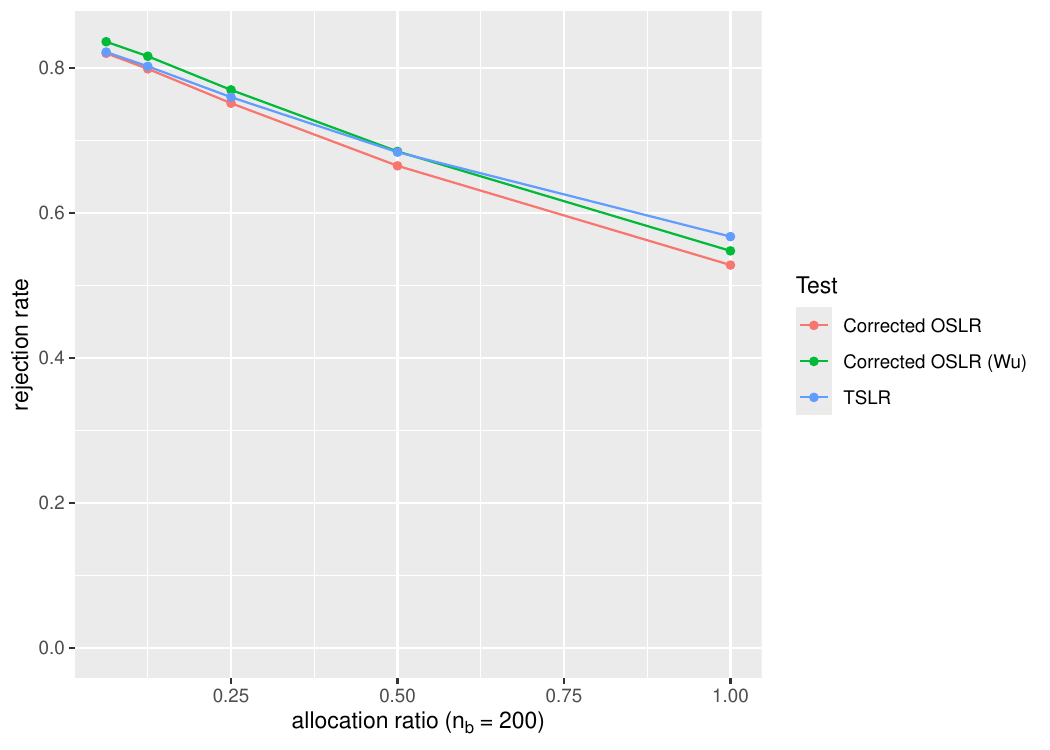"} &\includegraphics[width=.44\textwidth]{"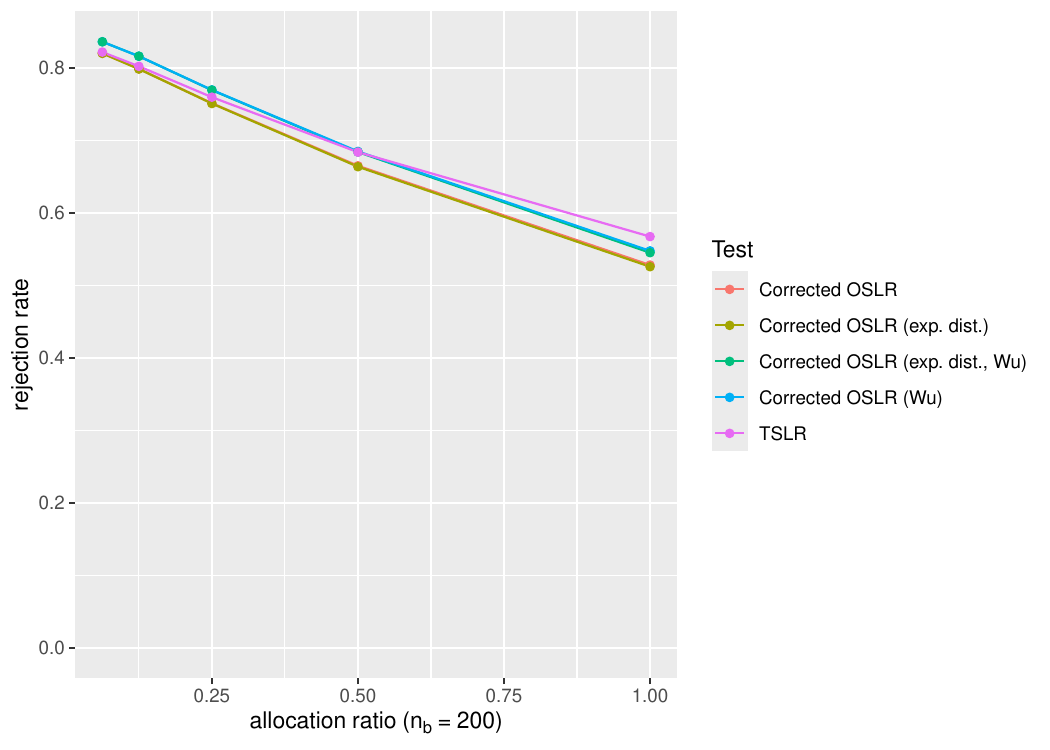"} 
	\end{tabular}
	\caption{Power to reject the null hypothesis $H_0$ for Weibull-distributed data with shape parameter $\kappa = 1$ and hazard ratio $\omega=0.8$ for four different sample sizes in the experimental cohort in dependence of the allocation ratio. The right column also includes procedures based on exponential MLEs}
\end{figure}

\begin{figure}[h!]
	\centering
	\begin{tabular}{c|c}
		$n_B$ & power \\
		\hline
		25 &\includegraphics[width=.44\textwidth]{"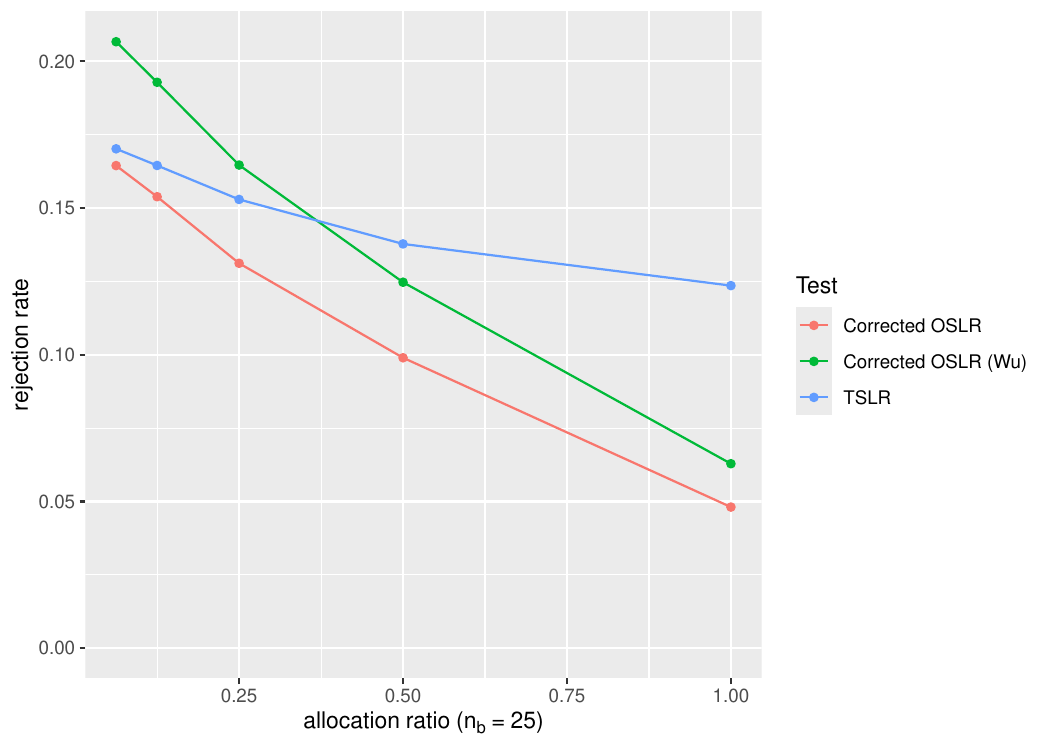"} \\
		\hline
		50 &\includegraphics[width=.44\textwidth]{"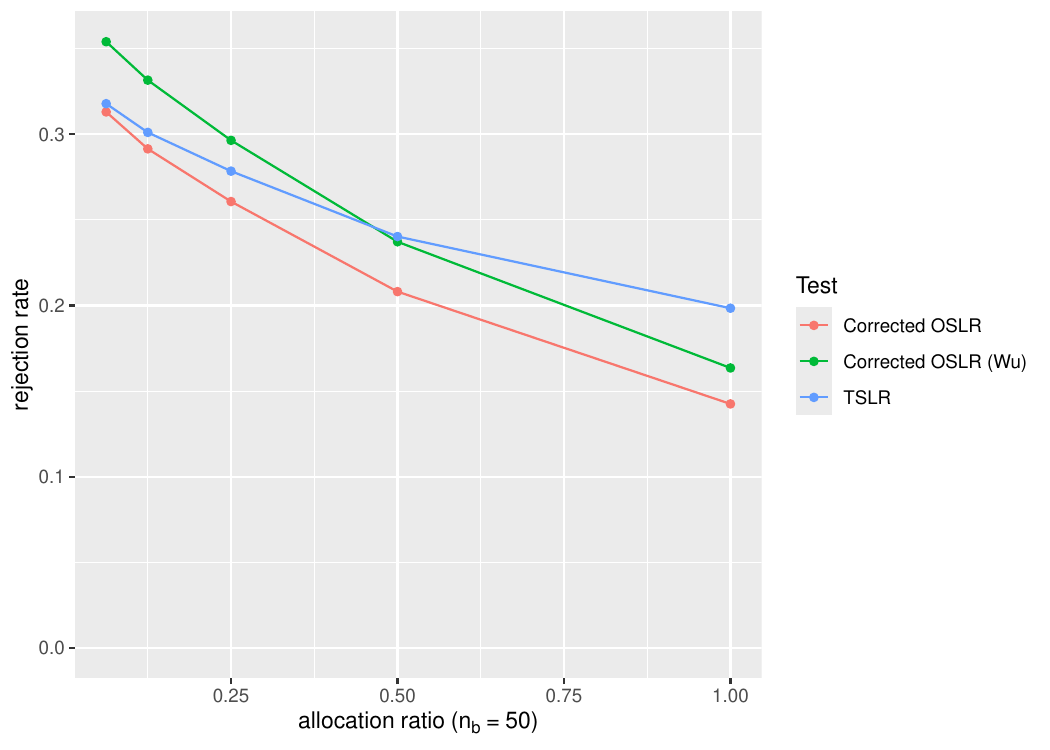"} \\
		\hline
		100 &\includegraphics[width=.44\textwidth]{"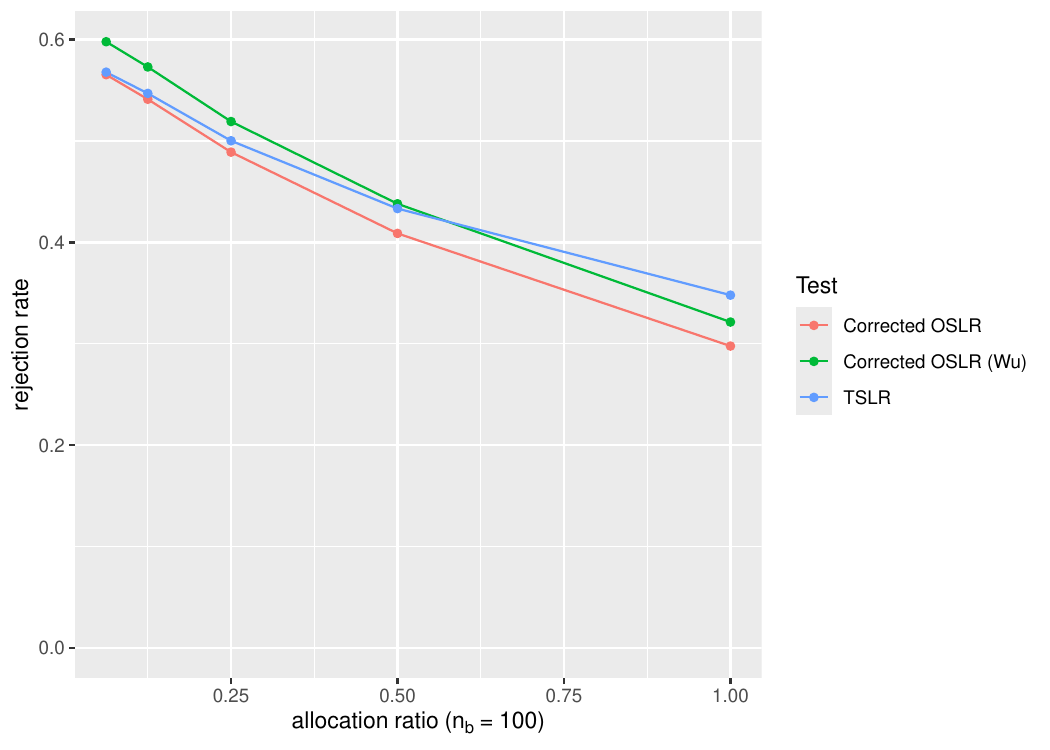"} \\
		\hline
		200 &\includegraphics[width=.44\textwidth]{"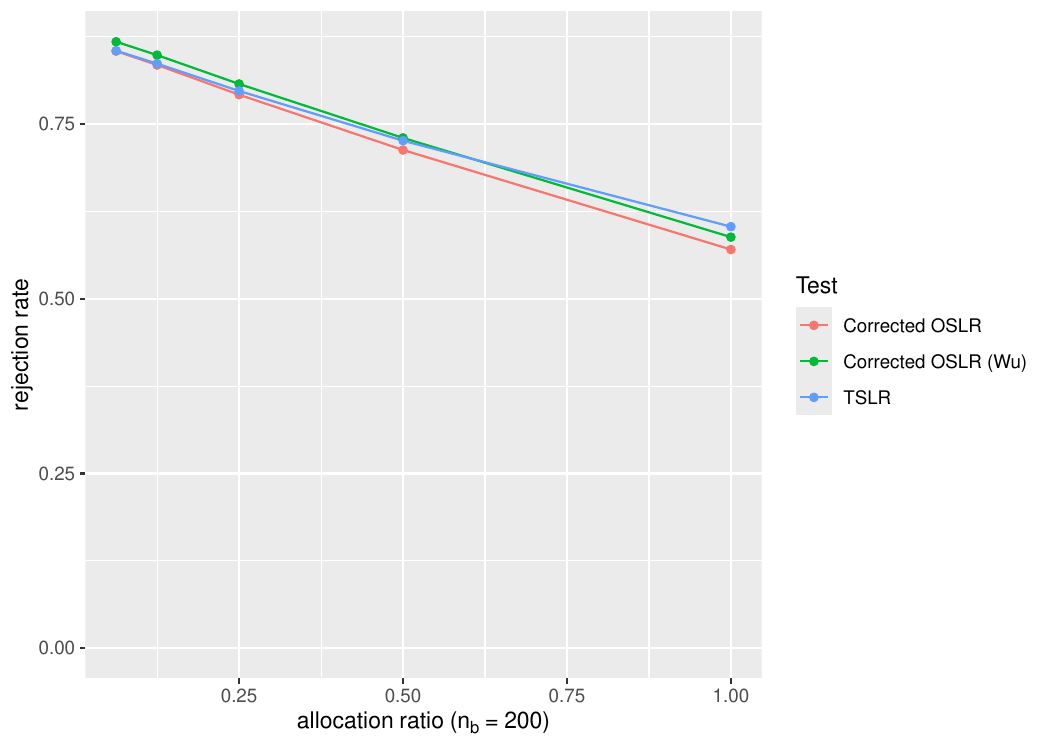"} 
	\end{tabular}
	\caption{Power to reject the null hypothesis $H_0$ for Weibull-distributed data with shape parameter $\kappa = 2$ and hazard ratio $\omega= 0.8$ for four different sample sizes in the experimental cohort in dependence of the allocation ratio.}
\end{figure}

\clearpage

\putbib
\end{bibunit}

\end{document}